\documentclass[lettersize,journal]{IEEEtran}
\usepackage{amsmath,amsfonts}
\usepackage{amsthm}
\usepackage{comment}
\usepackage{algorithmic}
\usepackage{algorithm}
\usepackage{array}
\usepackage[caption=false,font=normalsize,labelfont=sf,textfont=sf]{subfig}
\usepackage{textcomp}
\usepackage{stfloats}
\usepackage{url}
\usepackage{verbatim}
\usepackage{graphicx}
\usepackage{cite}
\usepackage{xurl}
\usepackage{xcolor}
\newtheorem{theorem}{Theorem}
\newtheorem{lemma}{Lemma}
\newtheorem{proposition}{Proposition}
\newtheorem{corollary}{Corollary}

\newtheorem{definition}{Definition}

\DeclareMathAlphabet{\mathbb}{U}{bbold}{m}{n}
\begin{document}

\title{To Save Mobile 
 Crowdsourcing from Cheap-talk: \\ A Game Theoretic Learning Approach}

\author{Shugang Hao,~\IEEEmembership{Member,~IEEE,} 
and~Lingjie~Duan,~\IEEEmembership{Senior Member,~IEEE}
        % <-this % stops a space
\thanks{Shugang Hao and Lingjie Duan are with the Pillar of Engineering Systems and Design, Singapore University of Technology and Design, Singapore, 487372 Singapore. E-mail: shugang\_hao@sutd.edu.sg, lingjie\_duan@sutd.edu.sg.

% Manuscript received 28 Sep. 2023; revised 16 Nov. 2023; accepted 28 Dec. 2023.
% Date of publication ; date of current version.

This work was supported by the Ministry of Education,
Singapore, under its Academic Research Fund Tier 2 under Grant Project
No. MOE-T2EP20121-0001. This work was also supported by the SUTD Kickstarter Initiative (SKI) Grant with project no. SKI 2021\_04\_07 and the Joint SMU-SUTD Grant with project no. 22-LKCSB-SMU-053. 

(Corresponding author: Lingjie Duan.)
 
}% <-this % stops a space
}



\maketitle

\begin{abstract}
Today mobile crowdsourcing platforms invite users to provide anonymous reviews about service experiences, yet many reviews are found biased to be extremely positive or negative. The existing methods find it difficult to learn from biased reviews to infer the actual service state, as the state can also be extreme and the platform cannot verify the truthfulness of reviews immediately. Further, reviewers can hide their (positive or negative) bias types and proactively adjust their anonymous reviews against the platform's inference.  
To our best knowledge, we are the first to study how to save mobile crowdsourcing from cheap-talk and strategically learn from biased users' reviews. 
We formulate the problem as a dynamic Bayesian game, including users' service-type messaging  and the platform's follow-up rating/inference. Our closed-form PBE shows that 
an extremely-biased user may still honestly message to convince the platform of listening to his review. Such Bayesian game-theoretic learning obviously outperforms the latest common schemes especially when there are multiple diversely-biased users to compete. For the challenging single-user case, we further propose a time-evolving mechanism with the platform's commitment inferences to ensure the biased user's truthful messaging all the time, whose performance improves with more time periods to learn from more historical data. 

% As such, we propose two new mechanisms beyond PBE: one on the time domain and the other on multi-user domain. Specifically, our time-evolving commitment mechanism ensures users' truthful messaging and improves with more time periods to learn historical data. While our multi-user mechanism exploits competition among diverse biased users and improves with user number.   
\end{abstract}


\section{Introduction}
\IEEEPARstart{T}{oday} mobile crowdsourcing platforms (e.g., TripAdvisor and Waze) invite users to submit anonymous reviews for rating their experienced services (e.g., of hotels, restaurants, and trips). Yet, many reviews are found biased to be extremely positive or negative. 
For example, 
one recent investigation reveals that anonymous users posted 79\% of the five-star fraudulent hotel reviews on TripAdvisor \cite{TA2}. As another example, carpet-cleaning company Hadeed was targeted by many anonymous negative reviews on Yelp \cite{Y}.  
 Given escalating criticisms of biased reviews in crowdsourcing platforms \cite{hbr}, it is critical for crowdsourcing platforms to strategically learn from these biased reviews to best infer the actual service state. However, the existing methods find it difficult to learn from biased reviews. On the platform side, it is challenging for the platform to verify users' reviews or even identities promptly, especially when the platform faces feedback on real-time service (e.g., Waze for navigation \cite{WZ-STA}).  Besides, extremely positive and negative reviews can also be the truth and they appear to be the majority among many reviews (e.g., \cite{hbr,hu2009online,gg}).  
 % Also,  the platform (e.g., TripAdvisor and Yelp) does not know the (positive or negative) bias types of anonymous users, 
 
  On the user side, major platforms' privacy protection policy allows reviewers to hide their identities and past reviews (implying their positive or negative bias types)  from a platform. For instance, TripAdvisor allows reviewers to maintain anonymity to feedback on its platform\cite{TA-4}, and Waze drivers can hide their identities when posting real-time traffic updates on the live map \cite{WZ-LM}. Such anonymity may skew the platform's final rating. Further, users are smart to adjust their reviews in proactive ways that mislead the platform's optimal service inference or recommendation \cite{TA-4}.
  For example,
suburban-area residents in the traffic navigation platform Waze tend to send fake reports about a speed trap, a wreck, and some other traffic snarl during rush hours to deflect the traffic near their homes  \cite{Waze1}. Besides, some policemen purposely posted no congestion messages in Waze to attract drivers there to catch speeders \cite{Waze3}.  
 These biased users behave very differently from malicious data attackers in the literature of crowdsourcing systems, who simply send fake reports to maximally undermine systems' inference accuracy (e.g., \cite{james2020sybil, tahmasebian2020crowdsourcing, zhao2023data,wang2020truth}). Crowdsourcing systems thus cannot learn from malicious attackers and just abandon their reviews in the final rating. In contrast, here a biased user aims to selfishly mislead the system to his preferred state and may still honestly reveal his preferred state in some cases. It is thus meaningful for the platform to strategically learn from such biased users.   \cite{li2022harnessing,yuan2020distributed} studied learning from crowdsourcing workers with uncertain information yet ignored the consideration of their biases.

Our problem to save crowdsourcing is more related to cheap-talk games in the economics or algorithmic game theory literature (e.g., 
% \cite{noukhovitch2021emergent, battaggion2022bright, li2008two, crastr1982,   karakocc2021cheap, ShiCoo2017, AmbAlm2014, McGChe2013}).
% \cite{noukhovitch2021emergent, battaggion2022bright, gick2006two, li2008two, crastr1982, li2004disclose, li2008mandatory, karakocc2021cheap, Vijasy2001, batmul2002, GilAsy1989, Vijmod2001, ShiCoo2017, AmbAlm2014, McGChe2013, AmbMul2008, PetWai2010, KarMul2013})
 \cite{battaggion2022bright, crastr1982,  karakocc2021cheap,  McGChe2013, ShiCoo2017 }), in which information senders (i.e., users in our problem) observe the nature state and send messages to proactively affect the receiver's (platform's)  inference of the actual state. This literature largely assumes that users have limited biases and just prefer the receiver to take inference slightly away from the actual state realization.
There are only a few recent works to analyze and tackle extreme bias (e.g., \cite{bhattacharya2018optimality,boleslavsky2016evolving}). However, unlike crowdsourcing, they only consider a simple scenario of one or two users, and strongly assume that a user's bias type is fixed and publicly known to aid the platform's state learning. In practice, a user to post anonymous reviews may have a positive or negative bias type, which is unknown to the platform. He may be bribed to target at over-the-top praise for business owners \cite{TA2}, or message poor rating to maximally attack service competitors or deflect the traffic near home \cite{Waze1}. Thus, prior game theoretic techniques cannot be applied to our problem to negate cheap-talk.

We summarize our key novelty and main results below. 
\begin{itemize}
    \item \emph{Strategic crowdsourcing to learn from biased users:} To our best knowledge, we are the first to study how to save mobile crowdsourcing from cheap-talk and strategically learn from 
    biased users. 
    Unlike the cheap-talk literature (e.g., \cite{battaggion2022bright, crastr1982,  karakocc2021cheap,  McGChe2013, ShiCoo2017, bhattacharya2018optimality,boleslavsky2016evolving}), our model is not limited to one or two users. We formulate the problem as a dynamic Bayesian game, including the service type messaging from users with hidden biases and the platform’s follow-up inference from the messages.  We also allow the users' biases to be extremely positive or negative, and practically consider that users' extreme biases are private information which are hidden from the platform. We present new analytical studies to provide guidance on \textit{how a platform can best learn useful information from biased users' reviews to provide service rating}.
    
\item \emph{Bayesian game theoretic approach to negate cheap-talks:} 
% We first show that at benchmarks in the literature (e.g., blind abandoning or majority-voting), the platform's loss can be arbitrarily large in the worst case. To reduce the platform's loss, w
We formulate a two-stage dynamic Bayesian game including an arbitrary number of users to
crowdsource service reviews for the platform's inference. In Stage I, users experiencing the service type send unverifiable messages on type realization to the platform, and then in Stage II the platform takes all the users' messages to best infer the actual service state. 
% As users' messaging have many ways to interplay and mislead the platform's inference from their observed state-distribution to their privately biased ones, the Bayesian game theoretic analysis becomes involved. Nonetheless, 
Our PBE is in closed-form and shows that the platform’s Bayesian game theoretic learning can greatly reduce biased users’ cheap-talks, where a user with extreme bias may still honestly message to
convince the platform of listening to his  review. Perhaps surprisingly, he may even message his non-preferred type to convince the platform.
 Such Bayesian game theoretic learning obviously outperforms the latest common schemes in the
literature (e.g., majority-voting and blind abandoning) especially
when there are multiple users of diverse biases to compete.

% \item \emph{New competition mechanism design for multiple diverse biased users:} To effectively learn from multiple biased users, we allow the platform to ask for their observations and strategically take recommendation actions, which involves competition among biased users to persuade the platform on their own messages. 
% % We extend our PBE analysis to an arbitrary number $N$ of users with hidden bias from the platform. 
% % In the two-user case, at the closed-form PBE we analytically show that the platform can still prevent cheap-talk in all the cases, and we prove that the platform's loss $\frac{(\mu_L-\mu_H)^2}{4}$ still holds in the worst case. Perhaps surprisingly, the platform's expected cost may be worse off after adding one more user. For arbitrary number of users, 
% We prove that either-biased users may honestly message to the platform at the PBE, and the platform loss is also greatly reduced compared to benchmarks in the literature. As user number becomes large, the platform tends to incur no loss for improvement. 

\item \emph{Time-evolving inference mechanism for truthful messaging:}
Though Bayesian game theoretic learning from multiple users, the system may still suffer non-trivial efficiency loss if the user number is small or the gap between high- and low-quality service types' distributions is large.
To further save crowdsourcing from cheap-talk and reduce system loss even in the challenging scenario of one available user, we design the platform's time-evolving commitment mechanism to enable the biased user's truthfully messaging all the time. Though this design problem has high dimensionality in the time domain, we manage to solve everything in closed-form.  Its performance improves if there are more time periods to learn from historical data, or the gap between high- and low-quality service types' distributions enlarges. Thus, it is beneficial for the platform to use the time-evolving mechanism than the game theoretic learning when the user number is small or the time period number is large.

% \item \emph{Extensions on sequentially-messaging users and user's asymmetrically-distributed biases:} We  extend our analysis on sequentially-messaging users with asymmetric messaging strategies. We prove that our analytical results on Bayesian learning approach still hold for this extended scenario. Finally, we extend to analyze PBE and system loss in the scenario of user's asymmetrically-distributed biases. 
 
% \item \emph{The platform's mechanism design to select users with noisy state observations:} To further save crowdsourcing from cheap-talk and reduce social efficiency loss, we design a novel mechanism for the platform to select users with noisy observations of the service state and avoid their cheap-talk. We also extend the binary state space to continuous. The selected users can only observe a binary signal as partial information on the service state, telling possibly positive or negative service state realization. We show that our mechanism can reduce cheap-talk at PBE, and $PoA$ also improves for arbitrary $N$ users.  

\end{itemize}

\begin{figure}
    \centering
\includegraphics[scale=0.47]{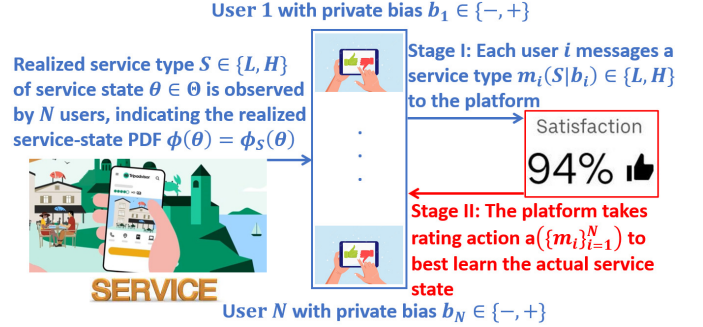}
    \caption{System model on $N$ users' messaging of observed service type $S$$\in$$\{L, H\}$ of service state $\theta$ (e.g., of hotels, restaurants and trips) to a crowdsourcing platform (e.g., TripAdvisor), who takes recommendation action $a(\{m_i\}_{i=1}^N)$ to best infer the actual service state. Here we consider the practice that the actual (continuous) service state $\theta$ cannot be accurately estimated by users, and only its low- or high-quality type is observed by each user. }
    \label{Fig.1}
\end{figure}

The rest of this paper is organized as follows. In Section II, we will introduce the system model and the dynamic Bayesian game formulation. In Section III, we will analyze two benchmark schemes including majority-voting and blind abandoning in the literature. In Section IV, we will analyze the Bayesian game theoretic learning approach and examine its system loss for an arbitrary number of biased users. In Section V, we will propose the time-evolving mechanism for the platform to incentivize truthful messaging from the biased user. In Section VI, we will extend our analysis to asymmetrically-distributed biases for users. Section VII concludes the paper.

\section{System Model and Problem Formulation}

\subsection{System Model of Crowdsourcing from Biased Users } 

As shown in Figure~\ref{Fig.1}, we consider a typical two-stage crowdsourcing scenario: $N$ users observe service type $S$$\in$$\{L, H\}$ (low- or high-quality) from their experiences and send review messages in Stage I, and then in Stage II the platform makes a service rating or quality inference for the recommendation. 

We first introduce some basics of the system model. Similar to \cite{boleslavsky2016evolving} and \cite{chakraborty2010persuasion}, 
% \footnote{We allow each user's observation to be noisy and will further study its impact in Section III.E.}  
we model that actual service state $\theta$ may vary in a continuous range $\Theta$$\subseteq$$R$, following a general probability density function (PDF) $\phi(\theta)$$\in$$\{\phi_L(\theta)$, $\phi_H(\theta)\}$. In practice, users cannot accurately estimate the actual (continuous) service state $\theta$. Instead, they can only observe the realized distribution type $S$$\in$$\{L, H\}$: low-quality type $S$=$L$ with  PDF function $\phi(\theta)$=$\phi_L(\theta)$ and high-quality type $S$=$H$ with PDF $\phi(\theta)$=$\phi_H(\theta)$ (e.g., \cite{li2019recommending,amin2018evaluating}).\footnote{We can similarly extend our PBE and system loss analysis to allow another medium service type $M$ in the distribution space $ \{\phi_L(\theta),  \phi_M(\theta), \phi_H(\theta)\}$. We prove that PBE.1-2 in Table~\ref{t2-p6} still occur at the PBE and the system loss is still reduced as the user number $N$ increases. Please refer to Appendix L for details.}  Further, we consider the challenging scenario that $\phi_L(\theta)$ and $\phi_H(\theta)$ have the same support $\Theta$, creating the maximal possible ambiguity to the platform's inference from the users' messages. This reflects the practice that another customer may still find a highly-rated hotel's service unsatisfactory, and might like to stay in a poorly-rated hotel \cite{2020Comprehending}.   
Unlike users, the platform has no observation on the high or low service type, and just knows the public probability distribution of $Pr(\phi(\theta)$$=$$\phi_H(\theta))$$=$$p_H$ and $Pr(\phi(\theta)$$=$$\phi_L(\theta))$$=$$1$$-$$p_H$ beforehand. Mean and variance of either PDF $\phi_j(\theta)$ are known to the platform as  $E[\theta|\phi(\theta)$=$\phi_j(\theta)]$=$\mu_j$ and $V[\theta|\phi(\theta)$=$\phi_j(\theta)]$=$\sigma_j^2$, respectively, where $j$$\in$$\{H, L\}$ and we assume $\mu_H$$\geq$$\mu_L$ and $\sigma_H^2$$\ne$$\sigma_L^2$ without loss of generality.

We then illustrate the two-stage interaction between $N$ users and the platform. In Stage I, after observing the type $S$ in $\{L, H\}$ of service state $\theta$, each user $i$$\in$$\{1,\cdots,N\}$  with bias type $b_i$  sends the platform unverifiable type message $m_i(S|b_i)$$\in$$\{L,H\}$. His private bias $b_i$ is realized from set $\{-, +\}$, telling negative or positive bias to mislead the platform from actual PDF $\phi(\theta)$. The platform only knows the prior probability distribution of possible biases as $Pr(b_i$=$+$)=$q_+$ and $Pr(b_i$=$-$)=1$-q_+$ with 0$<$$q_+$$<$1 for each user. Denote set $B$ to contain all the users' possible bias combinations and set $M$ to contain all the users' possible message combinations. In Stage II, after receiving $N$ users' messages $\{m_i\}_{i=1}^N$$\in$$M$, the platform takes a continuous inference action $a$ about the final service rating to best tell the actual service state.  Note that the platform's action $a$ is not necessarily equal to mean $\mu_H$ or $\mu_L$,  but is continuous to be any linear combination (e.g., 94\% satisfaction in Figure~\ref{Fig.1}).

Based on the above definitions, we next present each user's utility and the platform's cost functions, respectively. According to the cheap-talk literature (e.g., \cite{boleslavsky2016evolving,chakraborty2010persuasion,emons2019strategic}) and the fact that a positively/negatively-biased user expects the platform's action/rating as high/low as possible, user $i$'s utility depends on his bias $b_i$ to reflect self-preference for the platform's recommendation action $a(\{m_i\}_{i=1}^N)$ as follows:
% \begin{align}\label{u_s_}
%     &u_{S_i}(a(\{m_i\}_{i=1}^N), b_i) = a(\{m_i\}_{i=1}^N) \cdot b_i, \ b_i \in \{-1, 1\}.
% \end{align}
\begin{align}\label{u_s_}
    &u_{i}(a(\{m_i\}_{i=1}^N), b_i) = b_i \cdot a(\{m_i\}_{i=1}^N), \ b_i \in \{-, +\}.
\end{align}
The linear function above tells a user's negative or positive preference/bias on the platform's recommendation, reflecting gain (or loss) on service rating for a positively (or negatively)-biased user. 
% In practice, some customers' rates are found to be extreme in TripAdvisor, where some customers are bribed to post extremely positive rates to support, or extremely negative rates to attack (e.g., \cite{TA1,TA2}). 
As the platform's action strategy $a(\{m_i\}_{i=1}^N)$ is an inference from all the users' messages and observations, user $i$'s expected utility depends on the actual PDF of service state and the other users' possible messages to affect the platform's inference. It is given by: 
 \begin{align}
    \bar{u}(b_i) \!\!=\!\! \int_{\theta \in \Theta} & \bigg( \sum_{j \in \{L, H\} }  \!\!\!\!\!\!Pr(\phi(\theta)\!\! =\!\! \phi_j(\theta)) \!\!\!\!\!\!\!\!\sum_{\{m_i\}_{i=1}^N \in M}\!\!\!\!\!\!\!\! Pr(\{m_i\}_{i=1}^N|\phi_j(\theta)) \nonumber  \\
    % \end{align}
    %  \begin{align}
     &\ \ \ \ \ \ \ \ \ \ \ \ \ \ \ \ u_{i}(a(\{m_i\}_{i=1}^N), b_i) \phi_j(\theta) \bigg) d\theta. \label{u-s}
\end{align}
% Similar to the literature of cheap-talk (e.g., \cite{boleslavsky2016evolving,chakraborty2010persuasion,emons2019strategic}), we consider that u
User $i$ participates in rating in the long run and aims for maximizing his time-average utility in \eqref{u-s} over all possible PDF realizations.

Following the cheap-talk literature (e.g., \cite{crastr1982,karakocc2021cheap,ShiCoo2017}) and the crowdsourcing literature (e.g., \cite{tripkovic2021cluster,amini2013crowdlearner}), the platform's cost measures the square distance or error from its inference action $a(\{m_i\}_{i=1}^N)$ and realized state $\theta$, and is given as follows:
\begin{align}\label{u_r_}
    c(a(\{m_i\}_{i=1}^N), \theta) = (a(\{m_i\}_{i=1}^N) - \theta)^2, 
\end{align}
and its expected cost under a given inference strategy is thus
\begin{align}
    &\bar{c}(a(\{m_i\}_{i=1}^N)) \!\!=\!\! \int_{\theta\in \Theta} \bigg( \sum_{j \in \{L, H\} }  \!\!\!\!\!\!Pr(\phi(\theta)\!\! =\!\! \phi_j(\theta))  \nonumber \\
% \end{align}
% \begin{align}
     &\!\!\!\!\!\!\!\!\sum_{\{m_i\}_{i=1}^N \in M}\!\!\!\!\!\!\!\! Pr(\{m_i\}_{i=1}^N|\phi_j(\theta)) c(a(\{m_i\}_{i=1}^N), \theta) \phi_j(\theta) \bigg) d\theta. \label{u-r}
\end{align}

\subsection{Dynamic Bayesian Game Formulation}

We are now ready to formulate a two-stage dynamic Bayesian game including $N$ crowdsourcing users and the platform in Figure~\ref{Fig.1} as follows:
\begin{itemize}
    \item Stage I: each user $i$$\in$$\{ 1, 2, \cdots, N\}$ with private bias $b_i$$\in$$\{-, +\}$ observes the service type $S$$\in$$\{L, H\}$ of service state $\theta$ and simultaneously messages $m_i(S|b_i)$$\in$$\{L, H\}$ to the crowdsourcing platform to maximize his expected utility in \eqref{u-s}.\footnote{Our analysis for PBE and system loss performance can be extended to the case of sequentially-messaging users, which will be presented in Appendix K. }  The users with service experiences have information advantages over the platform to know the realized type of service state and their own bias types. 
    \item Stage II: after observing users' messages $\{m_i\}_{i=1}^N$, the platform chooses an inference action strategy $a(\{m_i\}_{i=1}^N)$ to minimize its expected cost in \eqref{u-r}.
\end{itemize}
Note that the platform's best strategy of inference action to minimize \eqref{u-r} is
\begin{align}
     &a^{*}(\{m_i\}_{i=1}^N)= E[\theta|\{m_i\}_{i=1}^N] \label{a} \\
     =& \arg\min_a \int_{\theta\in \Theta} \sum_{j \in \{L, H\} } Pr\big(\phi_j(\theta)|\{m_i\}_{i=1}^N\big)  c(a, \theta) \phi_j(\theta) d\theta. \nonumber 
\end{align}

Through Bayesian game theoretic analysis, we want to characterize the perfect Bayesian equilibrium (PBE) to proactively guide the platform's Bayesian game theoretic learning against any biased messaging.  { We then define PBE as follows.
\begin{definition}
A PBE is a set of strategies, $\{m_i^*\}_{i=1}^N$ and $a^*({\{m_i\}_{i=1}^N})$, with  beliefs $Pr(\phi(\theta)|\{m_i^*\}_{i=1}^N)$ such that
\begin{itemize}
    \item Each user $i$'s messaging $m_i^*(S|b_i)$ maximizes $\bar{u}(b_i)$ in \eqref{u-s}, given $\{m_j^*\}_{j \ne i}$, $a^*({\{m_i\}_{i=1}^N})$ and $Pr(\phi(\theta)|\{m_i^*\}_{i=1}^N)$ in \eqref{ul'-2} later.  
    \item The platform's inference $a^*({\{m_i\}_{i=1}^N})$ minimizes $\bar{c}(\cdot)$ in \eqref{u-r}, given $\{m_i^*\}_{i = 1}^N$ and $Pr(\phi(\theta)|\{m_i^*\}_{i=1}^N)$.

\end{itemize}
\end{definition}
}

% {\color{blue} We then define PBE as follows.
% \begin{definition}
% A PBE is a set of strategies, $\{m_i^*\}_{i=1}^N$ and $a^*{\{m_i\}_{i=1}^N}$, and beliefs $Pr(\theta|\{m_i^*\}_{i=1}^N)$ such that
% \begin{itemize}
%     \item $m_i^*(\theta|b_i)$ maximizes $\bar{u}_{S_i}(b_i)$ in \eqref{u-s}, given $\{m_j^*\}_{j \ne i}$, $a^*{\{m_i\}_{i=1}^N}$ and $Pr(\theta|\{m_i^*\}_{i=1}^N)$, where   $Pr(\theta|\{m_i^*\}_{i=1}^N)$ is obtained from \eqref{UL'}.
%     \item $a^*{\{m_i\}_{i=1}^N}$ maximizes $\bar{u}_{R}$ in \eqref{u-r}, given $\{m_i^*\}_{i = 1}^N$ and $Pr(\theta|\{m_i^*\}_{i=1}^N)$.
% \end{itemize}
% \end{definition}

% }

\subsection{Definition of Performance Metric: System Loss} 

Define $\bar{c}^{e,*}$ as the platform's expected cost at a particular PBE $e$ out of PBE set $E$ of the game, which has the form as \eqref{u-r} given the platform's recommendation action $a_e^{*}(\{m_i\}_{i=1}^N)$ in \eqref{a}. In the ideal case that the platform knows the actual PDF type (though not attainable), the minimal possible loss is obtained and we use it as a comparison with our learning approach later. Here, the platform cares about minimizing its expected cost according to the realized PDF, and its recommendation action is just the mean of distribution:
\begin{align*}
     a^{**} = \arg\min_a \int_{\theta\in \Theta} c(a,\theta) \phi_j(\theta) d\theta, j \in\{L, H\}
\end{align*}
given $\phi(\theta)$=$\phi_j(\theta)$ is realized.
After substituting it to \eqref{u-r}, the platform's minimum expected cost $\bar{c}^{**}$ is given by: 
\begin{align}
    \bar{c}^{**} = p_H\sigma_H^2+(1-p_H)\sigma_L^2. \label{eq:scN}
\end{align}

To evaluate the platform's inference error in our approach from this minimum, we finally define (expected) system loss $\bar{L}$ as the difference between the platform's expected costs $\bar{c}^{e,*}$ under the worst PBE $e$ out of set $E$ and $\bar{c}^{**}$ in \eqref{eq:scN} under the optimum. It is given by: 
\begin{align}
   \bar{L}(\mu_H,\mu_L,q_+,p_H) = 
   |\max_{e \in E} \bar{c}^{e,*} - \bar{c}^{**}|. \label{POA}
\end{align}
Note that a larger $\bar{L}$ implies a larger efficiency loss of our Bayesian game theoretic learning approach from the optimum.

\section{Two Benchmark Schemes: Majority-Voting and blind Abandoning}\label{w3:S3} 

Before analyzing our Bayesian game theoretic learning approach, in this section,  we first analyze two common benchmarks in the literature to be compared with our learning approach later.

\subsection{Benchmark 1: The Platform's Majority-Voting Scheme}

For the benchmark 1 in the crowdsourcing literature (e.g., \cite{tao2018domain,xu2018reward,9364734}), the platform simply favors the message with more users. Thus, each positively-biased user $i$ blindly messages high-PDF type (i.e., $m_i(S|b_i$=$+)$=$H$) and each negatively-biased user $j$ blindly messages low-PDF type (i.e., $m_j(S|b_j$=$-)$=$L$) to best mislead the platform. The platform's inference action is then
\begin{align*}
   &\bar{a}_1 = \\
   &\begin{cases}
      \arg\min_a \int_{\theta \in \Theta} c(a,\theta) \phi_H(\theta) d\theta, &\!\!\!\!\text{if} \ k \!>\! \frac{N}{2}, \\
     \arg\min_a \!\! \int_{\theta \in \Theta} \!\!\!\! \!\!\sum\limits_{j\in\{L, H\} }\!\!\!\!\!\!\!\!Pr(\phi(\theta) \!\! = \!\! \phi_j(\theta)) c(a, \theta) \phi_j(\theta) d\theta, &\!\!\!\!\text{if} \ k \!=\! \frac{N}{2}, \\
     \arg\min_a \int_{\theta \in \Theta} c(a,\theta) \phi_L(\theta) d\theta, &\!\!\!\!\text{if} \ k \!<\! \frac{N}{2},
   \end{cases} 
\end{align*}
where $k$ out of $N$ is the number of users messaging $m$=$H$. We can solve $\bar{a}_1$ because of $c(a,\theta)$'s convexity in $a$. Together with \eqref{u-r}, we have the following.
 \begin{lemma}\label{lemma-bm-2}
   At the benchmark 1 where the platform adopts the majority-voting scheme,  its recommendation action is 
    \begin{align}\label{a_1}
        \bar{a}_1 = \begin{cases}
            \mu_H, &\text{if} \ k > \frac{N}{2}, \\
            p_H\mu_H + (1-p_H)\mu_L, &\text{if} \ k = \frac{N}{2}, \\
            \mu_L, &\text{if} \ k < \frac{N}{2},
        \end{cases}
    \end{align}
    where $k$ is the number of users messaging $m=H$. The resultant system loss is
    \begin{align}
        &\bar{L}_1 = (\mu_H\!-\!\mu_L)^2\frac{1}{2^N}\bigg(p_H\!\!\sum_{k=0}^{\lceil\frac{N}{2}\!-\!1\rceil}C_N^k\! +\!(1\!-\!p_H) \!\!\!\sum_{l=\lfloor\frac{N}{2}\!+\!1\rfloor}^{N} \!\!\!C_N^l   \label{u-r-b-2} \\
        &\ \ \ \ \ \ \ \ \ \ \ \ \ \ \ \ \ \ \ \ \ \ \ \ \ \ \ \ \ \ \ \ +\text{$\mathbb{1}_{\lfloor\frac{N}{2}\rfloor=\frac{N}{2}}$}\cdot p_H(1\!-\!p_H)C_N^\frac{N}{2}\bigg), \nonumber
    \end{align}
which does not monotonically decrease with use number $N$. Further, we have $\bar{L}_1 > p_H(1-p_H)(\mu_H-\mu_L)^2$, which can be arbitrarily large for any $N$ as $\mu_H - \mu_L \to \infty$. 
\end{lemma}
The proof of Lemma~\ref{lemma-bm-2} is given in Appendix A. In \eqref{a_1} of the benchmark 1, the platform acts consistently as the mean of PDF type with more supporting messages, and takes action as the expected mean if high- and low-PDF messages are counted equal. Since each biased user blindly messages his biased PDF type to the platform, its rating action under the majority-voting mechanism can be quite different from the actual service state. Accordingly, its system loss $\bar{L}_1$ in \eqref{u-r-b-2} can be arbitrarily large as the mean difference of PDFs $\mu_H-\mu_L$ tends to infinity, regardless of the crowdsoucing effect from many users. Note that our analysis of the performance of the majority-voting mechanism is consistent with the literature and real-world examples (e.g., \cite{tullock1959problems,saunders2006majority,M1}).

\subsection{Benchmark 2: The Platform's Blind Abandoning Scheme}

For the benchmark 2 in the crowdsourcing literature (e.g., \cite{james2020sybil,  zhao2023data,wang2020truth}), the platform abandons biased anonymous reviews and never chooses to learn from them, which is also known as the ``babbling equilibrium" in the cheap-talk literature \cite{farrell1996cheap}. Thus, the platform wants to blindly minimize its expected cost according to its initial belief on PDF as follows:
\begin{align}\label{a-bm-1}
   \min_a \int_{\theta \in \Theta} \sum_{j\in\{ L, H\} }Pr(\phi(\theta) = \phi_j(\theta)) c(a, \theta) \phi_j(\theta) d\theta.
\end{align}
We can solve \eqref{a-bm-1} because of $c(a,\theta)$'s convexity in $a$.
Together with \eqref{u-r}, we have the following.

\begin{lemma}\label{lemma-bm-1}
    At the benchmark 2 where the platform abandons biased anonymous reviews, it reaches the  ``babbling equilibrium" to always choose recommendation action:
    \begin{align*}
        \bar{a}_2 = p_H\mu_H + (1-p_H)\mu_L. 
    \end{align*}
    The resultant system loss is
    \begin{align}\label{u-r-b-1}
        \bar{L}_2 = p_H(1-p_H)(\mu_H\!-\!\mu_L)^2,
    \end{align}
    which can be arbitrarily large as $\mu_H-\mu_L \to \infty$.
\end{lemma}

At the benchmark 2, the platform always takes the same fixed action, regardless of biased users' reviews. 

\begin{corollary}\label{C-1}
    The platform always incurs greater system loss at the benchmark 1 than 2.
\end{corollary}
% \begin{proof}
% To prove $\Bar{L}_1 > \Bar{L}_2$, we have 
% \begin{align*}
%     &\bar{L}_1 > (\mu_H\!-\!\mu_L)^2\bigg(p_H(1-p_H)\!\!\sum_{k=0}^{\lceil\frac{N}{2}\!-\!1\rceil}C_N^k\frac{1}{2^N} \\
%     &+\!p_H(1\!-\!p_H) \!\!\!\sum_{l=\lfloor\frac{N}{2}\!+\!1\rfloor}^{N} \!\!\!C_N^l\frac{1}{2^N} +\text{$\mathbb{1}_{\lfloor\frac{N}{2}\rfloor=\frac{N}{2}}$}\cdot p_H(1\!-\!p_H)C_N^\frac{N}{2}\frac{1}{2^N}\bigg) \\
%     &=p_H(1-p_H)(\mu_H\!-\!\mu_L)^2=\bar{L}_2. &\qedhere
% \end{align*}
% \end{proof}
As compared to the benchmark 2, at the benchmark 1 the platform's action may be quite different from the actual state since each biased user with majority advantage always messages his biased PDF type to the platform.  Thus, it incurs even greater loss than the benchmark 2. Like the benchmark 1, if the two PDFs' means are quite different, its fixed inference action incurs huge inaccuracy on actual state learning and its loss $\bar{L}_2$ in \eqref{u-r-b-1} can be arbitrarily large as the mean difference on PDFs $\mu_H-\mu_L$ tends to infinity. Therefore, we are well motivated to propose Bayesian game theoretic learning to improve from the next section. 

\section{Our Bayesian Game Theoretic Approach}\label{w3:S3-} 

In this section, we present our Bayesian game theoretic approach to negate cheap-talk for an arbitrary user number $N$. The platform asks for biased users' observations and strategically takes recommendation actions, which stimulates competition among biased users (of diverse biases) to persuade the platform of their own messages. Note that each user needs to take all the others' possible biases/messaging strategies into consideration when deciding his own messaging strategy, leading to more involved message combinations for the platform to strategically learn from. Here we adopt the widely used concept of PBE to analyze how best-acting users converge to a stable equilibrium. We first analyze each user's best strategy against any possible strategy combinations of the others. Then we combine their best response functions to find the fixed point as the stable PBE.

\subsection{Users' Strategy Simplification for Bayesian Game Theoretic Analysis}\label{w3:S3.1}

 The Bayesian game theoretic analysis is involved, as each user's messaging has multiple ways to mislead the platform's inference from his observed type to his privately biased one.  
% In the following we will find PBE first and then check the platform's expected cost gap from the optimum via analyzing $L$. 
Each user $i$ with bias $b_i$ and his observed service type $S$ will decide among four possible deterministic strategy candidates:
\begin{itemize}
    \item[1)] user $i$'s honest messaging strategy: 
    
    $m_i(S|b_i)$=$S$, $\forall S \in \{L, H\}$,
    \item[2)] user $i$'s blindly high-type messaging strategy: 
    
    $m_i(S|b_i)$=$H$, $\forall S \in \{L, H\}$,
    \item[3)] user $i$'s blindly low-type messaging strategy: 
    
    $m_i(S|b_i)$=$L$, $\forall S \in \{L, H\}$,
    \item[4)] user $i$'s reversed type messaging strategy:
    
    $m_i(S$=$H|b_i)$=$L$ and $m_i(S$=$L|b_i)$=$H$.
\end{itemize}

As in \cite{9364734, dipalantino2009crowdsourcing} and \cite{ghosh2012crowdsourcing}, here we reasonably consider users’ symmetric messaging strategy under the
same bias, i.e., $m_i(S_i|b_i) = m_j(S_j|b_j)$ if $b_i = b_j$ and
$S_i=S_j$.\footnote{Our analysis for PBE and system loss performance can be extended to the case of asymmetrically-messaging users, which will be presented in Appendix K. } 
Note that there are still 16 combinations of $N$ biased users' messaging strategies (either bias with 4 possible strategy candidates) even if we focus on users' symmetric messaging strategy under the same bias, yet our following lemma helps equivalently reduce the space to analyze users' equilibrium strategy.

\begin{lemma}\label{L11}
At the PBE of the dynamic Bayesian game with arbitrary $N$ users and the positive bias probability $q_+ \in (0, 1)$, only the following 6 out of the total 16 strategy combinations (SC) of $N$ users' messaging may occur:
\begin{itemize}
    \item[SC.1)] Blindly low-type messaging strategy for negatively-biased users (i.e., $m_i(S|b_i$=$-)$=$L$) and honest messaging strategy for positively-biased users  
    (i.e., $m_i(S|b_i$=$+)$=$S$, $\forall S \in \{L, H\}$). 

    \item[SC.2)] Honest messaging strategy for negatively-biased users (i.e., $m_i(S|b_i$=$-)$=$S$) and blindly high-type messaging strategy for positively-biased users (i.e., $m_i(S|b_i$=$+)$=$H$, $\forall S \in \{L, H\}$).
    
    \item[SC.3)] Blindly high-type messaging strategy for negatively-biased users (i.e., $m_i(S|b_i$=$-)$=$H$) and reversed messaging strategy for positively-biased users (i.e., $m_i(S$=$H|b_i$=$+)$=$L$ 
    and $m_i(S$=$L|b_i$=$+)$=$H$, $\forall S$$\in$$\{L, H\}$).

    % \item users' strategy combination 3. honest strategy for $b_i$=$-$ and blindly high strategy for $b_i$=$+$: 
    
    % $m_i(\phi(\theta)|b_i$=$-)$=$\phi(\theta)$ and $m_i(\phi(\theta)|b_i$=$+)$=$\phi_H(\theta)$.

   \item[SC.4)] Reversed messaging strategy for negatively-biased users (i.e., $m_i(S$=$H|b_i$=$-)$=$L$,  $m_i(S$=$L|b_i$=$-)$=$H$) and blindly low-type messaging strategy for positively-biased users (i.e., $m_i(S|b_i$=$+)$=$L$, $\forall S \in \{L, H\}$).
    
    % \item users' strategy combination 4. reversed strategy for $b_i$=$-$ and blindly low strategy for $b_i$=$+$: 
    % $m_i(\phi_H(\theta)|b_i$=$-)$=$\phi_L(\theta)$,  $m_i(\phi_L(\theta)|b_i$=$-)$=$\phi_H(\theta)$ and $m_i(\phi(\theta)|b_i$=$+)$=$\phi_L(\theta)$.
    
\item[SC.5)] Reversed messaging strategy for negatively-biased users (i.e., $m_i(S$=$H|b_i$=$-)$=$L$, $m_i(S$=$L|b_i$=$-)$=$H$) and honest messaging strategy for positively-biased users (i.e., $m_i(S|b_i$=$+)$=$S$, $\forall S \in \{L, H\}$).

    % \item users' strategy combination 5. reversed strategy for $b_i$=$-$ and honest strategy for $b_i$=$+$: 
    % $m_i(\phi_H(\theta)|b_i$=$-)$=$\phi_L(\theta)$,  $m_i(\phi_L(\theta)|b_i$=$-)$=$\phi_H(\theta)$ and $m_i(\phi(\theta)|b_i$=$+)$=$\phi(\theta)$.

\item[SC.6)] Honest messaging strategy for negatively-biased users (i.e., $m_i(S|b_i$=$-)$=$S$) and reversed messaging strategy for positively-biased users (i.e., $m_i(S$=$H|b_i$=$+)$=$L$ and $m_i(S$=$L|b_i$=$+)$=$H$, $\forall j \in \{L, H\}$).    
    % \item users' strategy combination 6. honest strategy for $b_i$=$-$ and reversed strategy for $b_i$=$+$: 
    
    % $m_i(\phi(\theta)|b_i$=$-)$=$\phi(\theta)$, $m_i(\phi_H(\theta)|b_i$=$+)$=$\phi_L(\theta)$  and $m_i(\phi_L(\theta)|b_i$=$+)$=$\phi_H(\theta)$.
\end{itemize}
\end{lemma}

% \begin{lemma}\label{L11}
% At the PBE of the dynamic Bayesian game with $N$ users, only the following six strategy combinations may occur:
% \begin{itemize}
%     \item[i)] user $i$'s strategy combination 1: blindly low strategy for $b_i$=$-1$ and honest strategy for $b_i$=$1$.
%     % \begin{align*}
%     %  m_i(\phi_i(\theta)|b_i=1)=\phi_i(\theta), \ m_i(\phi_i(\theta)|b_i=-1)=\phi_L(\theta).
%     % \end{align*}
%     \item[ii)] user $i$'s strategy combination 2: blindly high strategy for $b_i$=$-1$ and reversed strategy for $b_i$=$1$.
%     \item[iii)] user $i$'s strategy combination 3: reversed strategy for $b_i$=$-1$ and honest strategy for $b_i$=$1$.
%     \item[iv)] user $i$'s strategy combination 4: honest strategy for $b_i$=$-1$ and reversed strategy for $b_i$=$1$.
%     \item[v)] user $i$'s strategy combination 5: honest strategy for $b_i$=$-1$ and blindly high strategy for $b_i$=$1$.
%     \item[vi)] user $i$'s strategy combination 6: reversed strategy for $b_i$=$-1$ and blindly low strategy for $b_i$=$1$.
% \end{itemize}
% \end{lemma}

The proof of Lemma~\ref{L11} is given in Appendix B. Lemma~\ref{L11} helps reduce the number of users' strategy combination candidates from 16 to 6. It shows that at the PBE, the negatively-biased users' blindly messaging will not happen simultaneously with the positively-biased users' blindly messaging.  The platform can still learn useful information even from these extremely-biased users.

For ease of exposition, in the remaining of this section we first focus on the case of users' symmetrically-distributed biases, i.e.,  $q_+$=$\frac{1}{2}$,
% \footnote{ For the case of users' asymmetrically-distributed biases, it is more involved and not tractable to give closed-form PBE on arbitrary $N$ users. Nonetheless, as $N \to \infty$, we find that the platform incurs no loss in the most cases, which is consistent with Proposition~\ref{prop-11}.}
which will be extended to cover asymmetrically-distributed biases later in Section VI.
Denote set $I_j$=$\{i|$$m_i$=$j$,1$\leq$$i$$\leq$$N\}$ to include all users for messaging $j$ with $j$$\in$$\{L, H\}$. After receiving particular $|I_H|=k$  users' $H$ messages and $N-k$ users' $L$ messages, the platform updates its posterior state belief according to Bayes' theorem as follows:
\begin{align}\label{ul'-2}
    & Pr(\phi(\theta) \!=\! \phi_H(\theta)| |I_H|\!=\!k)  \\
    =& \frac{Pr(|I_H|=k|\phi(\theta) = \phi_H(\theta))Pr(\phi(\theta)=\phi_H(\theta))}{\sum\limits_{j\in\{L, H\}}\!\!\!\!\!Pr(|I_H|=k|\phi(\theta) = \phi_j(\theta))Pr(\phi(\theta)=\phi_j(\theta))}, \nonumber 
    % &= \frac{(Pr(m_i \!=\! \phi_H(\theta)|\phi(\theta) = \phi_H(\theta)))^k(Pr(m_{-i} \!=\! f_{L}|\phi(\theta) \!=\! \phi_H(\theta)))^{N\!-\!k}Pr(\phi(\theta)=\phi_H(\theta))}{\sum\limits_{l\in\{H,L\}}\!\!\!\!\!(Pr(m_i \!=\! f_l|f = f_l))^k(Pr(m_{-i} \!=\! f_{-j}|f \!=\! f_l))^{N\!-\!k}Pr(f=f_l)},  \nonumber
\end{align}
where we have
\begin{align*}
    &Pr(|I_H|\!\!=\!\!k|\phi(\theta) \!\!=\!\! \phi_H(\theta))\!\!  \\
    =& C_N^k (Pr(m_i \!\!=\!\! H|\phi(\theta) \!\!= \!\!\phi_H(\theta))^k 
     (Pr(m_{j} \!\!=\!\! L|\phi(\theta) \!\!=\! \!\phi_H(\theta)))^{N\!\!-\!\!k}
\end{align*}
and
the truthful reporting probability is
\begin{align}\label{ul-2}
   &Pr(m_i = H|\phi(\theta) = \phi_H(\theta))  \\
   =& \sum_{b \in \{-, +\}}\!\!\!\! Pr(m_i = H, b_i = b|\phi(\theta) = \phi_H(\theta)) \nonumber \\
   % \end{align}
   % \begin{align*}
  =& \sum_{b \in \{-,+\}} \!\!\!\! Pr(m_i = H|\phi(\theta) = \phi_H(\theta), b_i = b)Pr(b_i = b).\nonumber
\end{align}

\subsection{The Platform's Best-Response Inference Strategies to Users' Strategy Combinations SC.1-6}\label{w3:S3.2} 
Next, we examine the platform's best-response inference actions to the user's six possible strategy combinations, respectively, and will use it to proactively determine the users' strategy choice in Section IV-C. Assuming each user $i$ will adopt strategy combination SC.1 in Lemma~\ref{L11}, the platform expects: 
\begin{align}\label{ul1}
    &Pr(m_i = L|\phi(\theta) = \phi_j(\theta), b_i = - ) \\
   =&Pr(m_i = j|\phi(\theta) = \phi_j(\theta), b_i = + ) = 1, \ j \in \{L, H\}. \nonumber
\end{align}
Substituting \eqref{ul1} to \eqref{ul-2} and \eqref{ul'-2}, the platform's best-response inference action in \eqref{a} to strategy combination SC.1 is to infer
\begin{align}\label{a2-1}
    a_1^*(|I_H|=k) = 
    \begin{cases}
    \frac{p_H\frac{1}{2^N}\mu_H+(1-p_H)\mu_L}{p_H\frac{1}{2^N}+(1-p_H)}, &\!\!\text{if} \ k = 0, \\
        \mu_H, &\!\!\text{if} \ k \geq 1.     
    \end{cases}
\end{align}
Then the user's expected utility in \eqref{u-s} by using strategy combination SC.1 above is 
\begin{align}\label{u-1}
    &\bar{u}_1(b_i) = \nonumber \\
    &\begin{cases}
    -(p_H (\frac{1}{2})^{N - 1} + (1 - p_H))\frac{p_H(\frac{1}{2})^N\mu_H+(1-p_H)\mu_L}{p_H(\frac{1}{2})^N+(1-p_H)} & \\
    -p_H (1 - (\frac{1}{2})^{N - 1}) \mu_H, &\!\!\!\!\!\!\!\!\!\!\!\!\!\!\!\!\!\!\!\!\!\!\!\!\text{if} \ b_i \!= \!\!-,  \\
    p_H \mu_H \!+\! (1\!-\!p_H)\frac{p_H(\frac{1}{2})^N\mu_H+(1-p_H)\mu_L}{p_H(\frac{1}{2})^N+(1-p_H)}, &\!\!\!\!\!\!\!\!\!\!\!\!\!\!\!\!\!\!\!\!\!\!\!\!\text{if} \ b_i \!=\! +. 
    \end{cases}
\end{align}

Similar to the best response analysis to users' strategy combination SC.1, we obtain the platform's best-response actions to strategy combinations 2-6 in the following, respectively, 
\begin{align}
% \end{align}
% \begin{align}
    &a_2^*(|I_H|=k) = \label{a2-3}
    \begin{cases}
     \frac{p_H\mu_H + (1-p_H)\frac{1}{2^N}\mu_L}{p_H + (1-p_H)\frac{1}{2^N}}, &\text{if} \ k = N, \\
    \mu_L, &\text{if} \ k < N.   
    \end{cases} \\
% \end{align}
% \begin{align}
&a_3^*(|I_H|=k) = \label{a2-2}
    \begin{cases}
    \frac{p_H\frac{1}{2^N}\mu_H+(1-p_H)\mu_L}{p_H\frac{1}{2^N}+(1-p_H)}, &\!\!\text{if} \ k \!=\! N, \\
    \mu_H, &\!\!\text{if} \ k \!<\! N.    
    \end{cases} \\
% \end{align}
% \begin{align}
    &a_4^*(|I_H|=k) = \label{a2-4} 
    \begin{cases}
    \mu_L, &\text{if} \ k \geq 1, \\
     \frac{p_H\mu_H + (1-p_H)\frac{1}{2^N}\mu_L}{p_H + (1-p_H)\frac{1}{2^N}}, &\text{if} \ k = 0.       
    \end{cases} \\
% \end{align}
% \begin{align}
    &a_5^*(|I_H|=k) = p_H \mu_H+(1-p_H)\mu_L,  \label{a2-5} \\
% \end{align}
% \begin{align}
    &a_6^*(|I_H|=k) = p_H \mu_H+(1-p_H)\mu_L.  \label{a2-6} 
\end{align}

\subsection{Final Messaging Strategy and Inference Choice at PBE}\label{w3:S3.3} 
After comparing each user $i$'s expected utilities 
% in \eqref{u-1} and \eqref{u-3}-\eqref{u-6-b+1}
under the platform's best-response inference actions in \eqref{a2-1} and \eqref{a2-3}-\eqref{a2-6} for different biases $b_i$=$-$ and $b_i$=$+$,
respectively, we are able to determine his strategy for any bias type and finalize PBE of the whole game.

\begin{table}[t]
\caption{PBE for our Bayesian game theoretic approach for arbitrary $N$ users.}
\label{t2-p6}
% \vskip 0.15in
\begin{center}
\begin{small}
\begin{sc}
\begin{tabular}{|l|l|}
\hline
4 PBEs & Description  \\
\hline
\hline
 PBE.1 & $b_i$=- users: blindly low-type messaging \\
 (with  & $m_i^*(S|b_i \!= \!-) \!= \!L$, $\forall S \in \{L, H\}$.\\
 \cline{2-2}
SC.1 in  & $b_i$=+ users: honest messaging \hfill  \\
Lemma~\ref{L11})  & $m_i^*(S|b_i \!=\! +) \!=\! S$, $\forall S \in \{L, H\}$. \\
  \cline{2-2}
 & The Platform's inference 
 $a^*(m)$ in \eqref{a2-1}. \\
 % and honest messaging for $b_i$=+:\\
 % $m_i^*(\phi_j(\theta)|b_i \!= \!-) \!= \!\phi_L(\theta)$, \\ 
 % $m_i^*(\phi_j(\theta)|b_i \!=\! +) \!=\! \phi_j(\theta)$, $\forall j \in \{L, H\}$; \\
 % and the platform's inference \\
 % $a^*(m)$ in \eqref{a2-1}.  \\

\hline
 PBE.2 &  $b_i$=- users:  honest messaging  \hfill \\
 (with &  $m_i^*(S|b_i \!= \!-) \!= \!S$, $\forall S \in \{L, H\}$.\\
 \cline{2-2}
 SC.2 in  &$b_i$=+ users: blindly high-type messaging   \\
Lemma~\ref{L11})  & $m_i^*(S|b_i \!=\! +) \!=\! H$, $\forall S \in \{L, H\}$. \\
  \cline{2-2}
 & The Platform's  inference 
 $a^*(m)$ in \eqref{a2-3}. \\

 \hline
 PBE.3 & $b_i$=- users: blindly high-type messaging  \\
(with   &  $m_i^*(S|b_i \!= \!-) \!= \!H$, $\forall S \in \{L, H\}$.\\
 \cline{2-2}
 SC.3 in & $b_i$=+ users: reversed messaging   \\
Lemma~\ref{L11})  & $m_i^*(S\!=\!L|b_i \!= \!+) \!= \!H$, $m_i^*(S\!=\!H|b_i \!=\! +) \!=\! L$. \\
  \cline{2-2}
 & The Platform's  inference 
 $a^*(m)$ in \eqref{a2-2}. \\
 
\hline

PBE.4 & $b_i$=- users: reversed messaging  \\
(with & $m_i^*(S\!=\!H|b_i \!= \!-) \!= \!L$, $m_i^*(S\!=\!L|b_i \!= \!-) \!= \!H$. \\
\cline{2-2}
SC.4 in &   $b_i$=+ users:  blindly low-type messaging   \\
Lemma~\ref{L11})  & $m_i^*(S|b_i \!=\! +) \!=\! L$, $\forall S \in \{L, H\}$. \\
  \cline{2-2}
 & The Platform's  inference 
 $a^*(m)$ in \eqref{a2-4}. \\
 
 \hline

\end{tabular}
\end{sc}
\end{small}
\end{center}
% \vskip -0.1in
\end{table}

\begin{proposition}\label{prop-11}
 Given an arbitrary user number $N$$\geq$1, there are in total 4 PBEs in closed-form as summarized in Table~\ref{t2-p6}. There our Bayesian game theoretic approach manages to persuade positively-biased users' honest messaging in PBE.1 and negatively-biased users' honest messaging in PBE.2.
\end{proposition}

The proof of Proposition~\ref{prop-11} is given in Appendix C. Comparing with Lemma~\ref{L11}, we find that only strategy combinations SC.1-4 appear at the PBE in Table~\ref{t2-p6}. A positively-biased user does not exhibit SC.5. The reason is that he cannot convince the platform of adopting his preferred high-PDF type even by honestly messaging, since a negatively-biased user disturbs to mislead the platform by reversely messaging. Similarly, SC.6 is symmetric to SC.5 and does not appear at the PBE.

Now let us explain the insights behind PBE.1-4 in Table~\ref{t2-p6}. At PBE.1 with users' SC.1, it is natural for each $b_i$=- user to blindly message his preferred low type. This, however, motivates the $b_i$=+ users to honestly message both high- and low-type realizations. Whenever the platform receives a message of high type, it is the truth and the platform trusts the messaged type. Thus, the $b_i$=+ users convince the platform to listen to them. Their honest messaging of unfavorable low-type does not harm themselves much,  since the platform may still infer high-type realization from the $b_i$=- users' blindly low-type messaging. Note that PBE.2 is just symmetric to PBE.1 to occur and has similar insights. 

We surprisingly note that at PBE.3 with users' SC.3, the $b_i$=- users blindly message the unfavorable high type and the $b_i$=+ users reversely message the observed type, which is quite counter-intuitive. This is because the $b_i$=- users compete with the $b_i$=+ users to message and persuade the platform. The $b_i$=- users pretend to be the $b_i$=+ ones by blindly messaging high type to disturb the real $b_i$=+ users' messaging to the platform.
% , which reduces the platform's trust on high-type realization when receiving high-type messages. 
Now honest messaging strategy cannot help the $b_i$=+ users convince the platform of high-type realization, and the $b_i$=+ users balance the voices at the platform side by reversely messaging his observed type. This at least convinces the platform of high-type realization when receiving low-type messages only from the reversely-messaging $b_i$=+ users. Their reversed messaging of unfavorable low-type does not harm themselves much, since the platform may still infer high-type realization from the $b_i$=- users' blindly high-type messaging. Note that PBE.4 is just symmetric to PBE.3 to occur and has similar insights.

% Users' biases $b_i$=$-$ and $b_i$=$+$ offset each other for the platform in expected sense. Thus, either-biased users exhibit honest or reversed strategy to convince the platform of their biased PDF realizations to incur smaller cost. 

Given 4 PBEs in Proposition~\ref{prop-11}, we pick the worst one to examine the robust performance of our Bayesian game theoretic learning approach and analyze the platform's system loss $\bar{L}$ from the optimum in \eqref{POA}.

\begin{theorem}\label{prop-12}
Given arbitrary $N$ users, our Bayesian game theoretic approach in Table~\ref{t2-p6} has the system loss
\begin{align}\label{N}
    &\bar{L}_{Bayesian} =  \\
    &\max\bigg\{  \frac{(1-p_H)p_H(\mu_H-\mu_L)^2}{p_H+2^N(1-p_H)}, \frac{(1-p_H)p_H(\mu_H-\mu_L)^2}{1-(1-2^N)p_H} \bigg\},   \nonumber
\end{align}
which decreases with the user number $N$ and has $\lim_{N\to\infty}\bar{L}_{Bayesian} = 0$.
 The incurred system loss is obviously less than $\bar{L}_1$ in \eqref{u-r-b-2} of the benchmark 1 and $\bar{L}_2$ in \eqref{u-r-b-1} of the benchmark 2. Particularly, its maximum loss reduction from $\bar{L}_2$ is $\frac{(\mu_H-\mu_L)^2}{8} \big(\frac{2^N \sqrt{2^N
   \left(2^N+8\right)}+3}{2^N-1}-2^N-5\big)$, which increases with the user number $N$.\footnote{According to Corollary~\ref{C-1}, the system loss of the benchmark 1 is always greater than that of the benchmark 2. We thus analyze our approach's loss reduction from the better benchmark 2.}
 % \begin{align*}
 %     \Bar{L}_2 - L_N \in \!\! \left[0, \frac{1}{8} \left(\frac{2^N \sqrt{2^N
 %   \left(2^N+8\right)}+3}{2^N-1}-2^N-5\right)\right].
 % \end{align*}
\end{theorem}

The proof of Theorem~\ref{prop-12} is given in Appendix D.  Note that our Bayesian game theoretic learning approach is an obvious improvement from the majority-voting scheme as the benchmark 1 in \eqref{u-r-b-2}, which cannot reach zero even crowdsourcing from infinitely many users. The reason is that our approach exploits 
 positively and negatively biased users to offset each other for the platform to learn useful information. As the user number $N$ increases to approach infinity, the platform learns the truth from their messages and achieves the optimum.

\section{New Time-evolving Commitment Mechanism Design for Truthful Messaging}\label{IV} 

According to Proposition~\ref{prop-11}, we find that either-biased users' blindly high- or low-type messaging strategy may still occur at the PBE. Besides, the system loss $\bar{L}_{Bayesian}$ in \eqref{N} can still be arbitrarily large if the mean difference of PDFs $\mu_H-\mu_L$ is huge and the user number N is small. This motivates us to further design a new mechanism to incentivize biased user's truthfully messaging and reduce the system loss (e.g., \cite{ye2020federated,li2021talk}). In this section, we consider the challenging scenario of one user with $N$=1.  We will first design truthful mechanisms for the platform to enable the user's honest messaging all the time, and then analyze the performance of the time-evolving mechanism to show the improvement against approaches mentioned in Sections III and IV.

\subsection{The Platform's One-Shot Commitment Mechanism Design} 
According to the Revelation Principle (e.g., \cite{boleslavsky2016evolving} and  \cite{gibbons1992game}), we focus on the platform's commitment mechanism for the user's truthfully messaging. We first look at the simplest one-shot messaging from the user and wonder if possible to design such a truthful mechanism. 
 
\begin{definition}[The Platform's One-Shot Commitment Mechanism]
The platform commits to inference action $a_j$ after observing $m = j$, $j \in \{L, H\}$. The case-dependent committed actions $a_L$ and $a_H$ should ensure that neither-biased user type obtains more expected utility by deviating from truthful messaging.
\end{definition}
How to determine the commitments $a_L$ and $a_H$ needs careful design. Unfortunately, we find that the above one-shot truthful (though exists) mechanism fails to improve from the benchmark 2 in Lemma~\ref{lemma-bm-1}.
\begin{lemma}\label{3-L}
    The platform's one-shot commitment actions for ensuring the user's truthful messaging are
    \begin{align*}
        a_L^* = a_H^* = p_H \mu_H + (1-p_H) \mu_L,
    \end{align*}
    which degenerates to the benchmark 2 and leads to the same poor system loss as $\bar{L}_2$ in \eqref{u-r-b-1}. 
\end{lemma}
The proof of Lemma~\ref{3-L} is given in Appendix E.  As it is hard to ensure high efficiency in one shot, we next design the platform's time-evolving mechanism in the long run.

\subsection{The Platform's Time-Evolving Commitment Mechanism Design and Performance Analysis}

To incentivize the user's truthful messaging and distinguish with the benchmark 2, we consider the platform's commitment actions over a finite number of time periods with a time horizon $T\geq2$ such that the user's observed PDF type does not change. According to the Revelation Principle, in the initial period 0, the user with either bias truthfully messages his observed PDF type to the platform with no deviation for greater utility. Service state $\theta_k$ following the same PDF type in each period $k \in \{1, \cdots, T\}$ may vary and is revealed after the platform takes commitment action $a_k$. In the beginning of each period $k \in \{1, \cdots, T\}$, the platform, having state observations $h_{k-1} = (\theta_1, \cdots, \theta_{k-1}) \in \Theta_{k-1}$ in the past $k-1$ periods, commits to action $a_j(h_{k-1})$ if $m = j$ in period 0, $j \in \{L, H\}$.  We formally define this mechanism in the following.
\begin{definition}[The Platform's Time-Evolving Commitment Mechanism]\label{D-2}
The platform designs its time-evolving mechanism in its interaction with the user over multiple $T$$\geq$2 periods as follows:
\begin{itemize}
    \item Initial Period $0$: the user with private bias $b \in \{-, +\}$ observes the realized PDF type $S \in \{L, H\}$ of service state $\theta$ and truthfully messages $m(S|b)$=$S$ to the platform, $S \in \{L, H\}$.\footnote{The platform's commitment mechanism can also be extended to the case where the user faces noise and imperfectly observes the realized PDF. We can still design its commitment actions similarly as in the clear PDF-observation case to enable the biased user's truthful messaging all the time.} 
    \item Period $k \in \{1, 2, \cdots, T\}$: the platform commits to an inference action $a_j(h_{k-1})$ in the beginning of period $k$ according to its past observations on service states $h_{k-1}$=$(\theta_1, \cdots, \theta_{k-1})\in\Theta_{k-1}$ and its received message $m=j$, $j\in\{L,H\}$. After that the actual service state $\theta_k$ in period $k$ is revealed to the platform. Its objective is to minimize its expected total inference error/cost over $T$ periods. 
    % with guarantee of no user's deviation from truthfully messaging
    
\end{itemize}
\end{definition}

According to Definition~\ref{D-2}, it is crucial to design the platform's time-evolving commitment actions $\{(a_L(h_{k-1}), a_H(h_{k-1}))\}_{k=1}^T$ to ensure that either-biased user type does not deviate from truthfully messaging of his observed PDF type. Similar to \eqref{u-s}, now the user with bias $b$ has expected utility per period according to the platform's state observation history $\{h_{k-1}\}_{k=1}^T$ as follows:
\begin{align}
    &\Bar{u}^T(b) = b \cdot \frac{1}{T} \sum_{k=1}^T  \int_{\Theta_{k-1}} \bigg(p_H \phi_H(h_{k-1}) a_H(h_{k-1}) \nonumber \\ 
    &+ (1-p_H)\phi_L(h_{k-1}) a_L(h_{k-1})\bigg) d h_{k-1}, b \in \{-, +\}, \label{u-s-m}
\end{align}
and the platform's expected cost per period is:
\begin{align}
    &\bar{c}^T = \frac{1}{T} \sum_{k=1}^T  \bigg(p_H \int_{\Theta_{k-1}}  \bigg( \int_{\theta \in \Theta}  (a_H(h_{k-1})-\theta)^2 \phi_H(\theta) d\theta \bigg) \nonumber \\
    &\phi_H(h_{k-1}) dh_{k-1} + (1-p_H) \int_{\Theta_{k-1}}  \bigg( \int_{\theta \in \Theta}  (a_L(h_{k-1})-\theta)^2 \nonumber  
\end{align}
\begin{align}
    &\ \ \ \ \ \ \ \ \ \ \ \ \ \ \ \ \ \ \ \ \ \ \ \ \ \ \ \ \ \ \phi_L(\theta) d\theta \bigg)\phi_L(h_{k-1}) dh_{k-1} \bigg),\label{u-r-md} 
\end{align}
where $h_{k-1} = (\theta_1, \cdots, \theta_{k-1}) \in \Theta_{k-1}$ denotes the platform's observed $k-1$ independent state realizations according to the PDF type distribution at the beginning of period $k$. $h_0$ denotes the null set with $\phi_H(h_0) = \phi_L(h_0) = 1$ by convention. 

Next, we design the commitment actions over time for the platform to ensure truthfulness and high efficiency, without knowing the user's bias type.  To ensure the positively-biased user's no deviation from messaging of true PDF to the other type, the platform needs to satisfy the following incentive-compatibility (IC) constraints according to \eqref{u-s-m} regarding low and high PDF realizations, respectively: 
\begin{align}
    &\sum_{k=1}^T \int_{\Theta_{k\!-\!1}} (a_L(h_{k\!-\!1})-a_H(h_{k\!-\!1})) \phi_L(h_{k\!-\!1}) dh_{k\!-\!1} \geq 0, \label{IC-1} \\
% \end{align}
% \begin{align}
    &\sum_{k=1}^T \int_{\Theta_{k\!-\!1}} (a_H(h_{k\!-\!1})-a_L(h_{k\!-\!1})) \phi_H(h_{k\!-\!1}) dh_{k\!-\!1} \geq 0. \label{IC-4}
\end{align}
Similarly, to ensure the negatively-biased user's no deviation from messaging of true PDF to the other, the platform needs to satisfy the following IC constraints according to \eqref{u-s-m} regarding low and high PDF realizations, respectively:
\begin{align}
    &\sum_{k=1}^T \int_{\Theta_{k\!-\!1}} (a_H(h_{k\!-\!1})-a_L(h_{k\!-\!1})) \phi_L(h_{k\!-\!1}) dh_{k\!-\!1} \geq 0, \label{IC-2} \\
% \end{align}    
% \begin{align}
    &\sum_{k=1}^T \int_{\Theta_{k\!-\!1}} (a_L(h_{k\!-\!1})-a_H(h_{k\!-\!1})) \phi_H(h_{k\!-\!1}) dh_{k\!-\!1} \geq 0. \label{IC-3}
\end{align}

The platform's objective is to minimize its expected total cost/error in \eqref{u-r-md} over $T$ periods. Therefore, the platform aims to find best $a_L(h_{k-1})$ and $a_H(h_{k-1})$ for each $k \in \{1, \cdots, T\}$ by solving the following minimization problem:
\begin{align}\label{Problem}
    \min_{\{a_L(h_{k-1})\}_{k=1}^N, \{a_H(h_{k-1})\}_{k=1}^N} &\eqref{u-r-md} \nonumber \\
s.t. \ &\eqref{IC-1}, \eqref{IC-4}, \eqref{IC-2}, \eqref{IC-3}.
\end{align}
Note that the objective \eqref{u-r-md} is convex in each $a_j(h_{k-1})$, and constraints \eqref{IC-1}-\eqref{IC-3} are linear in each $a_j(h_{k-1})$, $j$$\in$$\{L,H\}$. To help present our mechanism, we introduce three notions $\Lambda(\theta)$, $\alpha$ and $\beta$ as follows:
\begin{align*}
    \Lambda(\theta)\! =\! \frac{\phi_H(\theta)}{\phi_L(\theta)}, \ \alpha = \int_{\theta \in \Theta} \frac{\phi_L^2(\theta)}{\phi_H(\theta)}d\theta, 
     \beta = \int_{\theta \in \Theta} \frac{\phi_H^2(\theta)}{\phi_L(\theta)}d\theta.
\end{align*}
Though \eqref{Problem} is a convex problem, it is still not easy to solve given the high dimensionality in time and involved integrals in all the constraints. In spite of this, we manage to solve everything in closed-form.  

\begin{theorem}\label{prop-14}
Our time-evolving mechanism in Definition~\ref{D-2} for the platform can ensure either-biased user's truthful messaging, by choosing the following commitment actions $a_L^*(h_{k-1})$ and $a_H^*(h_{k-1})$, $k \in \{1, \cdots, T\}$:
\begin{align*}
    &a_L^*(h_{k-1}) = \mu_L + w(h_{k-1})\left(\mu _H-\mu
   _L\right)  > \mu_L, \\
   &a_H^*(h_{k-1}) = \mu_H - w(h_{k-1}) \frac{1-p_H}{\Lambda(h_{k-1}) p_H}\left(\mu _H-\mu
   _L\right) < \mu_H, 
\end{align*}
where 
\begin{align*} 
    &w(h_{k-1}) \!\!=\!\! \frac{\Lambda(h_{k-1}) \lambda_H + \lambda_L}{2(\mu _H-\mu
   _L)(1-p_H)}, s_{\alpha} \!\!= \!\!\sum_{k=1}^T\alpha^{k\!-\!1}, s_{\beta} \!\!= \!\!\sum_{k=1}^T\beta^{k\!-\!1}, \\
    &\lambda_L \!\!=\!\!  \frac{2 (1\!-\!p_H) p_H^2 (-\mu_H s_{\beta}
   T\!+\!\mu_H T^2\!+\!\mu_L s_{\beta}
   T\!-\!\mu_L T^2)}{p_H^2 (s_{\alpha}\!-\!T) (s_{\beta}\!-\!T)\!-\!p_H (s_{\alpha}
   (s_{\beta}\!-\!2 T)\!+\!T^2)\!+\!T(T\!-\!s_{\alpha})}, \\
   &\lambda_H \!\!=\!\! \frac{2 \left(1\!-\!p_H\right)^2 p_H (-\mu_H s_{\alpha} T\!+\!\mu _H T^2\!+\!\mu_L s_{\alpha} T\!-\!\mu_L T^2)}{p_H^2 (s_{\alpha}\!-\!T) (s_{\beta}\!-\!T)\!-\!p_H (s_{\alpha}
   (s_{\beta}\!-\!2 T)\!+\!T^2)\!+\!T(T\!-\!s_{\alpha})}. 
\end{align*}
The resultant system loss per period from the optimum is
\begin{align}\label{u-r-m}
    &\bar{L}_{evolving} =  \\ 
    &\frac{(1\!-\!p_H) p_H T (\mu _H\!-\!\mu
   _L)^2 ((p_H-1) s_{\alpha }\!-\!p_H
   s_{\beta }\!+\!T)}{p_H^2 (s_{\alpha
   }\!-\!T) (s_{\beta}\!-\!T)\!-\!p_H (s_{\alpha } (s_{\beta }\!-\!2 T)\!+\!T^2)\!+\!T
   (T\!-\!s_{\alpha})}\!, \nonumber
\end{align}
which is greatly reduced from $\Bar{L}_1$ in \eqref{u-r-b-2} of the benchmark 1 and $\Bar{L}_2$ in \eqref{u-r-b-1} of the benchmark 2.
\end{theorem}

The proof of Theorem~\ref{prop-14} is given in Appendix F. Intuitively, the platform takes the time-evolving action $a_L^*(h_{k-1})$ greater than $\mu_L$ when receiving the low PDF message and $a_H^*(h_{k-1})$ less than $\mu_H$ when receiving the high PDF message in any period $k$ to compensate the biased user in the expected sense in the long run. Thus, either-biased user type always truthfully messages his PDF observation. 

Without much loss of generality, we consider service state under each PDF type following normal distributions as $N(\mu_j, \sigma^2)$, where $j \in \{L, H\}$, to obtain clear engineering insights with a more explicit expression of $\bar{L}_{evolving}$ in \eqref{u-r-m}.

% \begin{corollary}\label{prop-15}
% Under the platform's commitment mechanism, we have its loss as follows:
% \begin{align}
%     &L_M(\mu_L, \mu_H, p_H, \sigma_H^2, \sigma_L^2) = \label{D}\\
%     & \frac{(1\!-\!p_H) p_H T  ((p_H-1) s_{\alpha}\!-\!p_H
%    s_{\beta}\!+\!T)(\mu _H\!-\!\mu
%    _L)^2}{p_H^2 (s_{\alpha}\!-\!T) (s_{\beta}\!-\!T)\!-\!p_H (s_{\alpha } (s_{\beta}\!-\!2 T)\!+\!T^2)\!+\!T
%    (T\!-\!s_{\alpha})} \nonumber \\
%    &< p_H(1-p_H)(\mu_H-\mu_L)^2, s_{\alpha} \!= \!\sum_{k=1}^T\alpha^{k\!-\!1}, s_{\beta} \!= \!\sum_{k=1}^T\beta^{k\!-\!1}.\nonumber
% \end{align}
%  $L_M$ decreases with $\alpha$ and $\beta$, respectively. As $T$$\to$$\infty$, we have $L_M$$\to$0.
% % (see Figure~\ref{POA-1} as illustration).
% \end{corollary}

% Corollary~\ref{prop-15} shows that the proposed mechanism reduces the platform's loss than benchmark 1. 

\begin{corollary}\label{coro-2}
    Under normally-distributed service state $\theta \sim N(\mu_j, \sigma^2)$, $j \in \{L, H\}$, the platform’s time-evolving mechanism has the following system loss from the optimum:
    \begin{align}\label{L-TE}
    &\bar{L}_{evolving}  \\
     =&\frac{1 }{ \frac{e^{T(\mu_H-\mu_L)^2/\sigma^2}-1}{T(\mu_H-\mu_L)^2(e^{(\mu_H-\mu_L)^2/\sigma^2}-1)} - \frac{1}{(\mu_H-\mu_L)^2}  + \frac{1/(p_H(1-p_H))}{(\mu_H-\mu_L)^2}},\nonumber
    \end{align}
    which decreases with the period number $T$ to learn from more historical data. Note $\lim_{T\to\infty}\bar{L}_{evolving} = 0$ and $\lim_{(\mu_H-\mu_L)\to\infty}\bar{L}_{evolving} = 0$, greatly improving from $\lim_{(\mu_H-\mu_L)\to\infty}\bar{L}_{Bayesian} = \infty$ for finite $N$ at the Bayesian game theoretic learning approach with $\bar{L}_{Bayesian}$ in \eqref{N},   and $\lim_{(\mu_H-\mu_L)\to\infty}\bar{L}_2 = \infty$ at the benchmark 2 with $\Bar{L}_2$ in \eqref{u-r-b-1}.
\end{corollary}

In our carefully designed time-evolving mechanism,  the platform takes the mean difference of PDFs $\mu_H-\mu_L$ into consideration when determining its evolving commitment actions $a_L^*(h_{k-1})$ and $a_H^*(h_{k-1})$ in Theorem~\ref{prop-14}. The system loss under the time-evolving mechanism thus vanishes instead of growing to infinity as in the Bayesian game theoretic learning approach for finite $N$ and the benchmark 2. Besides, as the period number $T$ becomes large, the platform gradually learns the real PDF with no loss.

\begin{figure}
    \centering
    \includegraphics[scale=0.335]{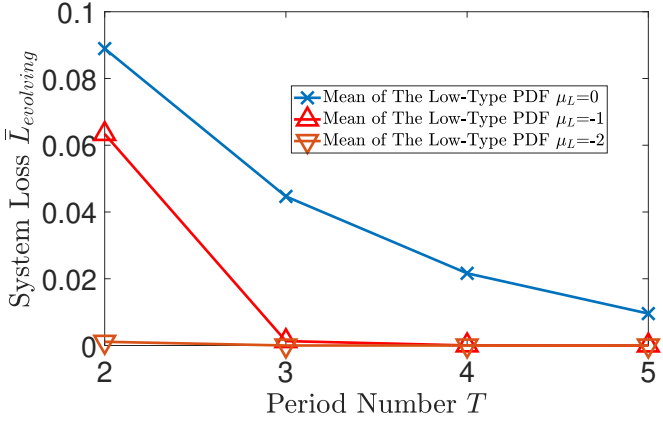}
    \caption{The system loss $\bar{L}_{evolving}$ in \eqref{u-r-m} of the time-evolving mechanism versus the period number $T$ and the mean of low-type PDF $\mu_L$. Here we set $p_H = 0.3$, choose $\phi_H(\theta)$ and $\phi_L(\theta)$ as normal distribution PDFs with $\mu_H$=1 and $\sigma^2_H$=$\sigma^2_L$=1.}
    \label{Fig.22-7}
\end{figure}

\begin{figure}
    \centering
    \includegraphics[scale=0.335]{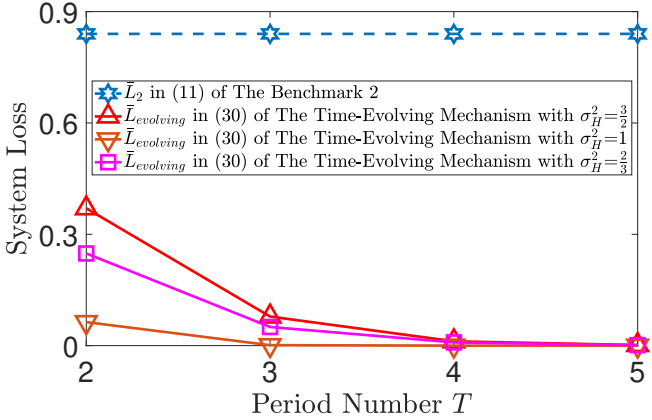}
    \caption{The system loss $\bar{L}_2$ in \eqref{u-r-b-1} of the benchmark 2 and $\bar{L}_{evolving}$ in \eqref{u-r-m} of the time-evolving mechanism, versus the high-PDF variance $\sigma_H^2$ and the period number $T$, respectively. Here we set $p_H = 0.3$, choose $\phi_H(\theta)$ and $\phi_L(\theta)$ as normal distribution PDFs with $\mu_H$=1, $\mu_L$=$-$1 and $\sigma^2_L$=1.}
    \label{Fig.22}
\end{figure}

 Figure~\ref{Fig.22-7} numerically shows that with the time-evolving mechanism, the system loss even decreases as the difference between means of PDFs $\mu_H-\mu_L$ increases, and it tends to the minimum under the time-evolving mechanism as the period number $T$ or $\mu_H-\mu_L$ increases, which is consistent with Corollary~\ref{coro-2}.

 \begin{figure}
    \centering
    \includegraphics[scale=0.335]{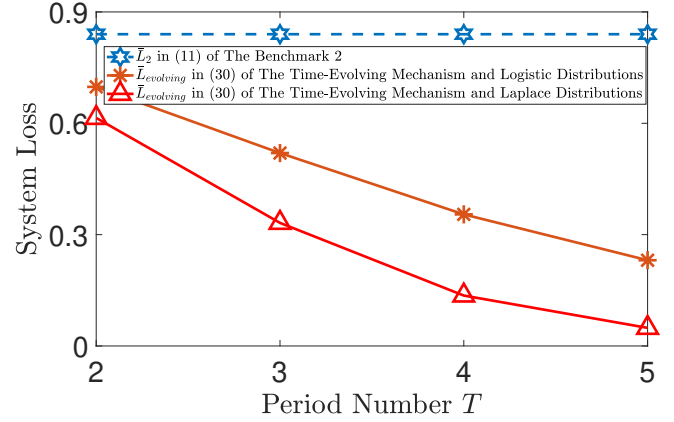}
    \caption{The system loss $\bar{L}_2$ in \eqref{u-r-b-1} of the benchmark 2 and $\bar{L}_{evolving}$ in \eqref{u-r-m} of the time-evolving mechanism, versus the period number $T$. Here we set $p_H = 0.3$, choose $\phi_H(\theta)$ and $\phi_L(\theta)$ as Logistic distribution PDFs Logistic($\mu$,$s$) with $\mu_H$=1, $\mu_L$=$-$1, $s_H=\frac{3}{2}$ and $s_L$=1,  and Laplace distribution PDFs Laplace($\mu$,$b$) with $\mu_H$=1, $\mu_L$=$-$1, $b_H=\frac{3}{2}$, $b_L$=1.}
    \label{Fig.22-6}
\end{figure}

Other than Corollary~\ref{coro-2}, one may wonder about the time-evolving mechanism's performance under asymmetric variance $\sigma_H^2 \ne \sigma_L^2$ for high- and low-type PDFs. Figure~\ref{Fig.22} numerically shows that given asymmetric variances $\sigma^2_H \ne \sigma^2_L$ of state distributions, the system loss $\bar{L}_{evolving}$ is still obviously smaller under the time-evolving mechanism than the benchmark 2, and it tends to zero as $T$ increases over 5,  which echoes with Corollary~\ref{coro-2}. 
As the variance $\sigma_H^2$ decreases from $\frac{3}{2}$ to $\sigma_H^2$=1, the two high- and low-type PDFs tend to vary similarly, making it easy for the platform to design commitment actions for satisfying the user's IC constraints \eqref{IC-1}-\eqref{IC-3}. Its loss is thus reduced. However, as the variance $\sigma_H^2$ keeps decreasing from 1 to $\frac{2}{3}$, the two PDFs tend to vary differently, making it hard for the platform to design commitment actions satisfying the user's IC constraints. Its loss is thus increased.

Note that our time-evolving mechanism in Theorem~\ref{prop-14} is designed for any state distribution and the system loss $\bar{L}_{evolving}$ in \eqref{u-r-m} is generally suitable for any distribution. Other than the normal distribution, Figure~\ref{Fig.22-6} numerically shows that our time-evolving mechanism's system loss under Laplace and Logistic distributions of service state, under which the system loss also tends to the optimum as the period number $T$ increases, which is consistent with Corollary~\ref{coro-2}.  Note that under Laplace and Logistic distributions, respectively,  $\bar{L}_{evolving}$ varies similarly versus the mean of low-type PDF $\mu_L$ as shown in Figure~\ref{Fig.22-7}, and also varies similarly versus variance-related parameters $b_H$ and $s_H$ of high-type PDF as shown in Figure~\ref{Fig.22}.

\subsection{Comparison with The Bayesian Game Theoretic Approach}

One may want to compare how the time-evolving mechanism performs with the Bayesian game theoretic learning approach for multiple users. After comparing the system loss in \eqref{L-TE} and \eqref{N}, we have the following.
\begin{proposition}\label{coro-3}
   Under normally-distributed service state $\theta \sim N(\mu_j, \sigma^2)$, $j \in \{L, H\}$, the platform incurs smaller system loss with the time-evolving mechanism than the Bayesian game theoretic learning approach from multiple users if and only if
   \begin{align*}
       N \leq \max \bigg\{& \log_2 \bigg( p_H \frac{e^{T(\mu_H-\mu_L)^2/\sigma^2}-1}{T(e^{(\mu_H-\mu_L)^2/\sigma^2}-1)} + (1-p_H)\bigg), \\
    % \end{align*}
    % \begin{align*}
       &\log_2 \bigg( (1-p_H) \frac{e^{T(\mu_H-\mu_L)^2/\sigma^2}-1}{T(e^{(\mu_H-\mu_L)^2/\sigma^2}-1)} + p_H\bigg) \bigg\},
   \end{align*}
   where the right-handed side of the above inequalities increases with $T$.
\end{proposition}
% As user number $N$ increases, the platform suffers from less loss with competition mechanism for multiple users. It is thus more likely to prefer asking multiple users than time-evolving against one.

\begin{figure}
    \centering
    \includegraphics[scale=0.335]{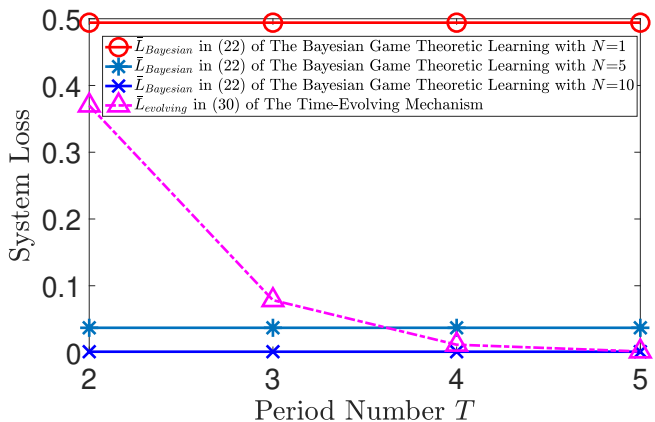}
    \caption{The system loss $\bar{L}_{Bayesian}$ in \eqref{N} of the Bayesian game theoretic learning approach and $\bar{L}_{evolving}$ in \eqref{u-r-m} of the time-evolving mechanism, versus the period number $T$ and the user number $N$, respectively. Here we set $p_H = 0.3$, choose $\phi_H(\theta)$ and $\phi_L(\theta)$ as normal distribution PDFs with $\mu_H$=1, $\mu_L$=-1, $\sigma^2_H$=$\frac{3}{2}$, and $\sigma^2_L$=1.}
    \label{Fig.23}
\end{figure}

Figure~\ref{Fig.23} numerically shows that the Bayesian game theoretic learning approach for multiple users tends to outperform the time-evolving mechanism only if $N$ becomes large, where the latter's performance improves with the period number $T$. This is consistent with Corollary~\ref{coro-3}.

\section{Extension to Users' Asymmetrically-Distributed Biases}

% \subsection{Extension on Sequentially-Messaging Users}

% In this subsection, we extend to consider the scenario where users sequentially message to the platform with asymmetric messaging strategies. For ease of exposition, we focus on the case of two users, where user 1 messages first and user 2 observing user 1's message then messages to the platform. Note that user 2 has more messaging strategies given his observation on user 1's message. We manage to prove that our PBE results in Proposition~\ref{prop-11} are stable with such sequentially-messaging users of asymmetric strategies.

% \begin{proposition}\label{prop-VC}
%     In the case of two sequentially-messaging users, cheap-talk never happens at the PBE. PBEs are stable to be the same as Table~\ref{t2-p6} for $N=2$. 
% \end{proposition}
% The proof of Proposition~\ref{prop-VC} is given in Appendix F of the supplementary material of this TMC submission. 
% Users' negative and positive biases well offset each other at the PBE for the platform to learn and act. Users 1 and 2 then exhibit symmetric strategy with same bias at the PBE.

Recall that we have assumed users' symmetrically-distributed biases with $q_+=\frac{1}{2}$ in Section IV. Now we extend to asymmetrically-distributed biases for further analysis. Note that users with different biases may exhibit different/asymmetric strategies due to unequal probability of appearance. The PBE analysis is more involved, depending on the distribution of the two biases. For ease of exposition, in the following we first look at a single user's case for deriving PBE in the closed-form under our Bayesian game theoretic learning approach, and then extend to arbitrary $N$ via numerical analysis. 

\begin{table}[t]
\caption{PBE versus high-type PDF probability $p_H$ in small positively-biased probability $q_+$$<$$\frac{3-\sqrt{5}}{2}$ in the one-user case.}
\label{t1:prop-1}
% \vskip 0.15in
\begin{center}
\begin{small}
\begin{sc}
\begin{tabular}{|l|l|}
\hline
\multicolumn{1}{|c|}{$p_H$ regime} & \multicolumn{1}{c|}{PBE results} \\
\hline
\hline
% & PBE \\
% \midrule
Small &  \multicolumn{1}{c|}{ PBE.2 (with SC.2 in Lemma~\ref{L11})} \\
% \cdashline{2-2}
$p_H \!\!\in\![0, p_{1, H}] $ & $b$=- user: honest messaging \\ 
& \multicolumn{1}{c|}{$m^*(S|b \!= \!-) \!= \!S$, $\forall S \in \{L, H\}$.}\\
% \cdashline{2-2}}
& $b$=+ user: blindly high-type messaging \\
& \multicolumn{1}{c|}{$m^*(S|b \!=\! +) \!=\! H$, $\forall S \in \{L, H\}$.} \\
% \cdashline{2-2}
% & $s \in \{-1, 1\}$;\\
& The Platform's inference  $a^*(m)$ in \eqref{a-3}.  \\
\cline{2-2} 
& \multicolumn{1}{c|}{ PBE.4 (with SC.4 in Lemma~\ref{L11})}  \\
% \cdashline{2-2}
 & $b$=- user: reversed messaging  \\
 & \multicolumn{1}{c|}{$m^*(S=H|b \!= \!-) \!= \!L$,}\\
& \multicolumn{1}{c|}{$m^*(S=L|b \!= \!-) \!= \!H$.} \\
% \cdashline{2-2}
 &$b$=+ user: blindly low-type messaging \\
& \multicolumn{1}{c|}{$m^*(S|b \!=\! +) \!=\! L$, $\forall S \in \{L, H\}$.} \\
% \cdashline{2-2}
&The Platform's inference $a^*(m)$ in \eqref{a-4}.  \\
\hline
Large 
& \multicolumn{1}{c|}{ PBE.1 (with SC.1 in Lemma~\ref{L11})}  \\
% \cdashline{2-2}
$p_H \!\!\in\![p_{1, H}, 1]$ & $b$=- user: blindly low-type messaging  \\
& \multicolumn{1}{c|}{$m^*(S|b \!= \!-) \!= \!L$, $\forall S \in \{L, H\}$.}\\ 
% \cdashline{2-2}
 &$b$=+ user: honest messaging\\
& \multicolumn{1}{c|}{$m^*(S|b \!=\! +) \!=\! S$, $\forall S \in \{L, H\}$.} \\
% & $s \in \{-1, 1\}$;\\
% \cdashline{2-2}
&The Platform's inference $a^*(m)$ in \eqref{a-1}.  \\
\cline{2-2} 
& \multicolumn{1}{c|}{ PBE.2 (with SC.2 in Lemma~\ref{L11})}  \\
\cline{2-2} 
& \multicolumn{1}{c|}{ PBE.3 (with SC.3 in Lemma~\ref{L11})}  \\
% \cdashline{2-2}
 &$b$=- user: blindly high-type messaging  \\
 & \multicolumn{1}{c|}{$m^*(S|b \!= \!-) \!= \!H$, $\forall S \in \{L, H\}$.}\\
% \cdashline{2-2}
 &$b$=+ user: reversed messaging \\
& \multicolumn{1}{c|}{$m^*(S=L|b \!= \!+) \!= \!H$,} \\
& \multicolumn{1}{c|}{$m^*(S=H|b \!=\! +) \!=\! L$.} \\
% \cdashline{2-2}
&The Platform's inference $a^*(m)$ in \eqref{a-2}.  \\
\cline{2-2}
& \multicolumn{1}{c|}{ PBE.4 (with SC.4 in Lemma~\ref{L11})} \\
\hline
\end{tabular}
\end{sc}
\end{small}
\end{center}
% \vskip -0.1in
\end{table}

\subsection{Bayesian Game Theoretic Learning Approach for $N$=1 in Closed-Form PBE}

In this subsection, we investigate the one-user case under any asymmetrically-distributed biases, i.e., $q_+$$\ne$$\frac{1}{2}$, and we skip the index in $b_i$ and $m_i$. By adding six possible bias-messaging combinations of the user to analysis in Section~\ref{w3:S3.2}, we can obtain the platform's best-response Bayesian game theoretic learning actions in \eqref{a} under strategy combinations SC.1-6 as follows: 
\begin{align}
&a_1^*(m)= \label{a-1}
\begin{cases}
        \mu_H, &\text{if} \ m = H, \\
        \frac{p_H(1-q_+)\mu_H+(1-p_H)\mu_L}{p_H(1-q_+)+(1-p_H)}, &\text{if} \ m = L.
    \end{cases} \\
% \end{align}
% \begin{align}
&a_2^*(m) = \label{a-3}
    \begin{cases}
     \frac{p_H\mu_H + (1-p_H)q_+\mu_L}{p_H + (1-p_H)q_+}, &\text{if} \ m = H, \\
    \mu_L, &\text{if} \ m = L.   
    \end{cases} \\
% \end{align}
% \begin{align}
&a_3^*(m) = \label{a-2}
    \begin{cases}
    \frac{p_H(1-q_+)\mu_H+(1-p_H)\mu_L}{p_H(1-q_+)+(1-p_H)}, &\text{if} \ m = H, \\
        \mu_H, &\text{if} \ m = L.
    \end{cases} 
 \end{align}
\begin{align}
    &a_4^*(m) = \label{a-4} 
    \begin{cases}
    \mu_L, &\text{if} \ m = H, \\
     \frac{p_H\mu_H + (1-p_H)q_+\mu_L}{p_H + (1-p_H)q_+}, &\text{if} \ m = L. 
     \end{cases} \\
% \end{align}
% \begin{align}
    &a_5^*(m) = \label{a-5}
    \begin{cases}
    \frac{p_H q_+ \mu_H+(1-p_H)(1-q_+)\mu_L}{p_H q_++(1-p_H)(1-q_+)}, &\text{if} \ m = H, \\
    \frac{p_H (1-q_+) \mu_H+(1-p_H)q_+\mu_L}{p_H (1-q_+)+(1-p_H)q_+}, &\text{if} \ m = L.
    \end{cases} \\
% \end{align}
% \begin{align}
    & a_6^*(m) = \label{a-6}
    \begin{cases}
     \frac{p_H (1-q_+) \mu_H+(1-p_H)q_+\mu_L}{p_H (1-q_+)+(1-p_H)q_+}, &\text{if} \ m = H, \\
    \frac{p_H q_+ \mu_H+(1-p_H)(1-q_+)\mu_L}{p_H q_++(1-p_H)(1-q_+)}, &\text{if} \ m = L.   
    \end{cases} 
\end{align}
 To help finalize the user's strategy at the PBE, we define two thresholds $p_{1, H}$$\leq$$1$ for $q_+$$<$$\frac{1}{2}$ and $p_{1, L}$$\geq$$0$ for $q_+$$>$$\frac{1}{2}$:
\begin{align}
    &p_{1,H} \! = \! \frac{3-q_+}{2} \!-\! \frac{1}{2} \sqrt{\frac{2 q_+^3-9
   q_+^2+12 q_+-5}{2 q_+-1}}, \ q_+ \!<\! \frac{1}{2}, \label{p1H} \\
 % \end{align} 
 % \begin{align}
    &p_{1,L}=\frac{1}{2} \sqrt{\frac{2 q_+^3+3 q_+^2}{2 q_+-1}}-\frac{q_+}{2}, \ q_+ > \frac{1}{2}. \label{p1L}
\end{align}
We can prove that the user with negative bias $b$=$-$ will exhibit honest strategy to convince the platform of no cheap-talk if $p_H$$<$$\min\{1,p_{1,L}\}$, and the user with negative bias $b$=$+$ will exhibit honest strategy to convince the platform of no cheap-talk if $p_H$$>$$\max\{0,p_{1,H}\}$. 

\begin{table}[t]
\caption{PBE in medium positively-biased probability $q_+$$\in$$[\frac{3-\sqrt{5}}{2}$,$\frac{\sqrt{5}-1}{2}$] in the one-user case.}
\label{t1:prop-2}
% \vskip 0.15in
\begin{center}
\begin{small}
\begin{sc}
\begin{tabular}{|l|l|}
\hline
4 PBEs & Description  \\
\hline
\hline
 PBE.1 & $b$=- user: blindly low-type messaging \\
 (with & $m^*(S|b \!= \!-) \!= \!L$, $\forall S \in \{L, H\}$.\\
 \cline{2-2}
SC.1 in & $b$=+ user: honest messaging \hfill  \\
Lemma~\ref{L11})  & $m^*(S|b \!=\! +) \!=\! S$, $\forall S \in \{L, H\}$. \\
  \cline{2-2}
 & The Platform's inference 
 $a^*(m)$ in \eqref{a-1}. \\
 % and honest messaging for $b_i$=+:\\
 % $m_i^*(\phi_j(\theta)|b_i \!= \!-) \!= \!\phi_L(\theta)$, \\ 
 % $m_i^*(\phi_j(\theta)|b_i \!=\! +) \!=\! \phi_j(\theta)$, $\forall j \in \{L, H\}$; \\
 % and the platform's inference \\
 % $a^*(m)$ in \eqref{a2-1}.  \\

\hline
 PBE.2 &  $b$=- user:  honest messaging  \hfill \\
 (with &  $m^*(S|b \!= \!-) \!= \!S$, $\forall S \in \{L, H\}$.\\
 \cline{2-2}
SC.2 in  &$b$=+ user: blindly high-type messaging   \\
Lemma~\ref{L11})  & $m^*(S|b \!=\! +) \!=\! H$, $\forall S \in \{L, H\}$. \\
  \cline{2-2}
 & The Platform's inference 
 $a^*(m)$ in \eqref{a-3}. \\

 \hline
 PBE.3 & $b$=- user: blindly high-type messaging  \\
(with  &  $m_i^*(S|b \!= \!-) \!= \!H$, $\forall S \in \{L, H\}$.\\
 \cline{2-2}
SC.3 in  & $b$=+ user: reversed messaging   \\
 Lemma~\ref{L11}) & $m^*(S=L|b \!= \!+) \!= \!H$, \\
  & $m^*(S=H|b \!=\! +) \!=\! L$. \\
  \cline{2-2}
 & The Platform's  inference 
 $a^*(m)$ in \eqref{a-2}. \\
 
\hline

PBE.4 & $b$=- user: reversed messaging  \\
(with & $m^*(S=H|b \!= \!-) \!= \!L$,\\
SC.4 in &  $m^*(S=L|b \!= \!-) \!= \!H$. \\
 \cline{2-2}
Lemma~\ref{L11}) & $b$=+ user:  blindly low-type messaging   \\
 & $m^*(S|b \!=\! +) \!=\! L$, $\forall S \in \{L, H\}$. \\
  \cline{2-2}
 & The Platform's  inference 
 $a^*(m)$ in \eqref{a-4}. \\
 
 \hline

\end{tabular}
\end{sc}
\end{small}
\end{center}
% \vskip -0.1in
\end{table}

\begin{proposition}\label{prop-1} 
If the user has small positively-biased probability (i.e., $q_+ < \frac{3-\sqrt{5}}{2}$), we have $p_{1,H}$$\in$$[0,1)$. Non-unique PBEs are given in closed-form in Table~\ref{t1:prop-1}, depending on high-type PDF probability $p_H$. 
\end{proposition}

The proof of Proposition~\ref{prop-1} is given in Appendix G. Note that PBE.1 and 3 only occur if high-PDF probability $p_H \geq p_{1,H}$ in Table~\ref{t1:prop-1} for small positively-biased probability $q_+ < \frac{3-\sqrt{5}}{2}$, which is different from Table~\ref{t2-p6} for symmetrically-distributed biases (with $q_+ = \frac{1}{2}$) to always occur. With small positively-biased probability $q_+$$<$$\frac{3-\sqrt{5}}{2}$, the platform expects the user with positive bias to appear less often. The positively-biased user type then aims to convince the platform of high-PDF type realization by honestly messaging only if the probability of high-PDF type is large enough (i.e., $p_H > p_{1, H}$). Otherwise, he blindly messages to best mislead the platform. Note that PBEs in the large positively-biased probability regime of $q_+ > \frac{\sqrt{5}-1}{2}$ are symmetric to those in this small regime and have similar insights. Details are given in Appendix I.

\begin{proposition}\label{prop-2} 
If the user has medium positively-biased probability (i.e., $q_+ \in [\frac{3-\sqrt{5}}{2}, \frac{\sqrt{5}-1}{2}]$), there are in total  4 PBEs in closed-form as summarized in Table~\ref{t1:prop-2}. 
% PBE.1-4 in Table~\ref{t2-p6} always occur for any $p_H \in [0, 1]$. 
\end{proposition}

The proof of Proposition~\ref{prop-2} is given in Appendix H. 
% Proposition~\ref{prop-2} implies that the negatively-biased user's cheap-talk (i.e., $m(\phi_j(\theta)|b$=$-)$=$\phi_L(\theta)$, $\forall j \in \{L, H\}$) and the positively-biased user's cheap-talk (i.e., $m(\phi_j(\theta)|b$=$+)$=$\phi_H(\theta)$, $\forall j \in \{L, H\}$) never occur as a strategy combination at the PBE.
% the biased user's cheap-talk does not happen at the PBE. 
With medium positively-biased probability $q_+$$\in$$[\frac{3-\sqrt{5}}{2}, \frac{\sqrt{5}-1}{2}]$, both biased user types appear with similar frequencies to exhibit PBE.1-PBE.4 as in the symmetrically-distributed biases case in Table~\ref{t2-p6}. Note that Proposition~\ref{prop-2} becomes the same as Proposition~\ref{prop-11} by letting $q_+$=$\frac{1}{2}$ and $N$=1.

Given PBE results in Propositions~\ref{prop-1}-\ref{prop-2}, we are ready to analyze the system loss in \eqref{POA} by comparing our Bayesian game theoretic learning performance to the optimum.

\begin{proposition}\label{prop-4}
In the one-user case, our Bayesian game theoretic learning approach incurs the system loss from the optimum as follows:
% \begin{align}\label{1}
%     &L_1(\mu_L, \mu_H, q_+, p_H) = (\mu_H-\mu_L)^2\\
%     &\max\bigg\{\frac{q_+ (1-p_H) p_H }{p_H + (1-p_H)q_+}, \frac{(1-q_+) (1-p_H) p_H }{1-q_+ p_H}\bigg\}. \nonumber
% \end{align}
\begin{align}\label{1}
    \bar{L}_{Bayesian} = \max\bigg\{  \frac{q_+ (1-p_H) p_H (\mu _H-\mu
   _L)^2}{p_H + (1-p_H)q_+}&, \nonumber \\
% \end{align} 
% \begin{align}
   \frac{(1-q_+) (1-p_H) p_H (\mu _H-\mu
   _L)^2}{1-q_+ p_H} &\bigg\},
\end{align}
% which decreases with user number $N$ and has $\lim_{N\to\infty}\bar{L}_{Bayesian} = 0$.
 % The incurred system loss 
 which is less than $\bar{L}_1$ in \eqref{u-r-b-2} of the benchmark 1 and $\bar{L}_2$ in \eqref{u-r-b-1} of the benchmark 2. Particularly, its maximum loss reduction from $\bar{L}_2$ is $(\sqrt{5}-2)(\mu_H-\mu_L)^2$, which occurs if $q_+ \to 0$ and $p_H = p_{1,H}$, or $q_+ \to 1$ and $p_H = p_{1,L}$.
\end{proposition}

The proof of Proposition~\ref{prop-4} is given in Appendix J.  
At the PBE.2 (or PBE.1), the negatively-biased (or positively-biased) user type aims to convince the platform of his preferred low-PDF (or high-PDF) type realization by truthfully messaging. As positively-biased probability $q_+$$\to$0 (or 1), his voice dominates to  occur almost all the time. Note that the system loss of the benchmark 2 is huge with medium high-PDF type probability $p_H = p_{1,H}$ (or $p_H = p_{1,L}$). System loss reduction is thus maximized at  $p_H = p_{1,H}$ (or $p_H = p_{1,L}$) and $q_+$$\to$0 (or 1) due to the negatively-biased (or positively-biased) user type's truthfully messaging. Note that Proposition~\ref{prop-4} becomes the same as Theorem~\ref{prop-12} by letting $q_+=\frac{1}{2}$ and $N$=1.

\begin{figure}
    \centering
    \includegraphics[scale=0.335]{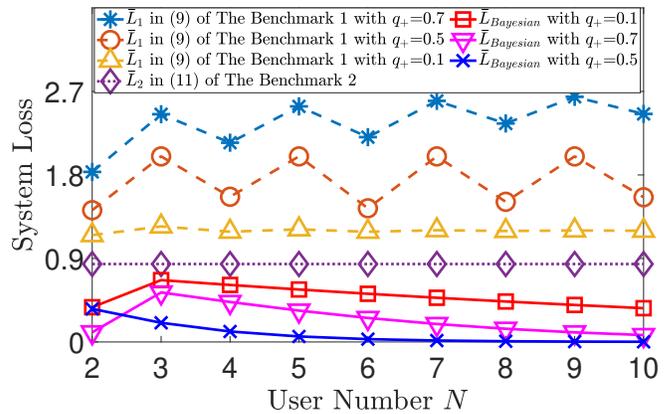}
    \caption{The system loss of the benchmarks 1-2 and the Bayesian game theoretic learning approach, versus the user number $N$ and the positively-biased probability $q_+$, respectively. Here we set $\mu_H$=1, $\mu_L$=$-$1, $\sigma_H^2$=$\frac{3}{2}$, $\sigma_L^2$=$1$, and $p_H = 0.3$.}
    \label{Fig.5}
\end{figure}

\subsection{Numerical Results on $N\geq2$ Users}

% \begin{proposition}\label{w3:c_r_}
In this subsection, we extend to the asymmetrically-distributed biases with $q_+ \ne \frac{1}{2}$ for arbitrary $N$ users. Lemma~\ref{L11} still holds for arbitrary $N\geq2$ users with any $q_+\in(0,1)$, and we can follow \eqref{ul'-2}-\eqref{ul-2} to obtain the platform's best-response learning actions in \eqref{a} under the possible strategy combinations 1-6. However, the analysis is more involved and we can still obtain the PBE numerically for arbitrary $N\geq2$ users.  

% We then numerically study the case of arbitrary $N$ users with asymmetrically-distributed biases. Denote system loss with our competition mechanism as 
% \begin{align}
%    \hat{L}_N(\mu_H,\mu_L,q_+,p_H) = 
%    |\max_{e \in E} \bar{u}^{e,*}_{R} - \bar{u}_R^{**}|. \label{L}
% \end{align}

Figure~\ref{Fig.5} numerically shows that with our Bayesian game theoretic learning approach, the system loss $\bar{L}_{Bayesian}$ is obviously smaller than those of benchmarks 1 and 2 with asymmetrically-distributed bias probability $q_+=0.1$ or 0.7. Our system loss may fluctuate as $N$ increases from 2 to 3, yet it tends to 0 as the user number $N$ keeps increasing, which is consistent with Theorem~\ref{prop-12}. Differently, the system loss of the benchmark 1 or 2 does not vanish with $N$, which keeps fluctuating or flat regardless of symmetrically- or asymmetrically-distributed biases $q_+$=0.1, 0.5, or 0.7.

\section{Conclusion}

In this paper, we study how to save mobile crowdsourcing from cheap-talk and strategically learn the actual service state from biased users' reviews. 
We formulate the problem as a dynamic Bayesian game, including the service type messaging from users with hidden biases and the platform's follow-up service rating/inference from the messages. Our PBE is in closed-form and shows that the platform's game theoretic learning can greatly reduce biased users' cheap-talk, where a user with extreme bias may still honestly message to convince the platform of listening to his review. Such Bayesian game theoretic learning obviously outperforms the latest common schemes in the literature (e.g., blind abandoning and majority-voting) especially when there are multiple users of diverse biases to compete. For the challenging single-user case, we further propose a time-evolving mechanism with the platform's commitment inferences to ensure the biased user's truthful messaging all the time, whose performance improves with more time periods to learn from more historical data.  

In the future, there are some directions to extend this work. For example, it may be interesting to consider an extension on time-varying service types, and to explore how biased users exhibit messaging strategies at the PBE and how the platform can incentivize their honest messaging. Further, it may also be interesting to consider the case of user collusion, which requires us to adopt new approximation methods to determine biased users' asymmetric messaging strategies within the exponentially increased space.

\section*{Acknowledgments}
A preliminary version of this paper was presented in \cite{hao23sav}.

% \bibliographystyle{IEEEtran}
% \bibliography{text}

\begin{thebibliography}{10}
\providecommand{\url}[1]{#1}
\csname url@samestyle\endcsname
\providecommand{\newblock}{\relax}
\providecommand{\bibinfo}[2]{#2}
\providecommand{\BIBentrySTDinterwordspacing}{\spaceskip=0pt\relax}
\providecommand{\BIBentryALTinterwordstretchfactor}{4}
\providecommand{\BIBentryALTinterwordspacing}{\spaceskip=\fontdimen2\font plus
\BIBentryALTinterwordstretchfactor\fontdimen3\font minus \fontdimen4\font\relax}
\providecommand{\BIBforeignlanguage}[2]{{%
\expandafter\ifx\csname l@#1\endcsname\relax
\typeout{** WARNING: IEEEtran.bst: No hyphenation pattern has been}%
\typeout{** loaded for the language `#1'. Using the pattern for}%
\typeout{** the default language instead.}%
\else
\language=\csname l@#1\endcsname
\fi
#2}}
\providecommand{\BIBdecl}{\relax}
\BIBdecl

\bibitem{TA2}
\BIBentryALTinterwordspacing
L.~Longwell. (2022, Apr.) What’s wrong with tripadvisor, and what to do about it. TravelAddicts. [Online]. Available: \url{https://traveladdicts.net/tripadvisor-fake-reviews/}
\BIBentrySTDinterwordspacing

\bibitem{Y}
\BIBentryALTinterwordspacing
D.~Post. (2015, april) Exposing anonymous yelp reviewers. The Washington Post. [Online]. Available: \url{https://www.washingtonpost.com/news/volokh-conspiracy/wp/2015/04/23/exposing-anonymous-yelp-reviewers/}
\BIBentrySTDinterwordspacing

\bibitem{hbr}
\BIBentryALTinterwordspacing
K.~et~al. (2018, Mar.) Online reviews are biased. here’s how to fix them. Harvard Business Review. [Online]. Available: \url{https://hbr.org/2018/03/online-reviews-are-biased-heres-how-to-fix-them}
\BIBentrySTDinterwordspacing

\bibitem{WZ-STA}
\BIBentryALTinterwordspacing
(2023, Jan) Waze statistics and user count (2023). Waze. [Online]. Available: \url{https://expandedramblings.com/index.php/waze-statistics-facts/#google_vignette}
\BIBentrySTDinterwordspacing

\bibitem{hu2009online}
N.~Hu, P.~A. Pavlou, and J.~J. Zhang, ``Why do online product reviews have a j-shaped distribution? overcoming biases in online word-of-mouth communication,'' \emph{Communications of the ACM}, vol.~52, no.~10, pp. 144--147, 2009.

\bibitem{gg}
\BIBentryALTinterwordspacing
D.~Schaal. (2023, Jan.) Google flights and hotels make europe-mandated transparency changes … in europe. Skift. [Online]. Available: \url{https://skift.com/blog/google-flights-and-hotels-make-europe-mandated-transparency-changes-in-europe/}
\BIBentrySTDinterwordspacing

\bibitem{TA-4}
\BIBentryALTinterwordspacing
Maintaining anonymity. TripAdvisor. [Online]. Available: \url{https://www.tripadvisorsupport.com/en-US/hc/traveler/articles/511}
\BIBentrySTDinterwordspacing

\bibitem{WZ-LM}
\BIBentryALTinterwordspacing
Live map. Waze. [Online]. Available: \url{https://www.waze.com/live-map/}
\BIBentrySTDinterwordspacing

\bibitem{Waze1}
\BIBentryALTinterwordspacing
L.~Vaas. (2016, Jun.) Waze to go: residents fight off crowdsourced traffic… for a while. Sophos. [Online]. Available: \url{https://nakedsecurity.sophos.com/2016/06/07/waze-to-go-residents-fight-off-crowdsourced-traffic-for-a-while/}
\BIBentrySTDinterwordspacing

\bibitem{Waze3}
\BIBentryALTinterwordspacing
(2015, Feb.) Are miami cops really flooding waze with fake police sightings? Mordor Intelligence. [Online]. Available: \url{https://www.mordorintelligence.com/industry-reports/location-based-services-market}
\BIBentrySTDinterwordspacing

\bibitem{james2020sybil}
J.~James, ``Sybil attack identification for crowdsourced navigation: A self-supervised deep learning approach,'' \emph{IEEE Transactions on Intelligent Transportation Systems}, vol.~22, no.~7, pp. 4622--4634, 2020.

\bibitem{tahmasebian2020crowdsourcing}
F.~Tahmasebian, L.~Xiong, M.~Sotoodeh, and V.~Sunderam, ``Crowdsourcing under data poisoning attacks: A comparative study,'' in \emph{Data and Applications Security and Privacy XXXIV: 34th Annual IFIP WG 11.3 Conference, DBSec 2020, Regensburg, Germany, June 25--26, 2020, Proceedings 34}.\hskip 1em plus 0.5em minus 0.4em\relax Springer, 2020, pp. 310--332.

\bibitem{zhao2023data}
Y.~Zhao, X.~Gong, F.~Lin, and X.~Chen, ``Data poisoning attacks and defenses in dynamic crowdsourcing with online data quality learning,'' \emph{IEEE Transactions on Mobile Computing}, vol.~22, no.~05, pp. 2569--2581, 2023.

\bibitem{wang2020truth}
Y.~Wang, K.~Wang, and C.~Miao, ``Truth discovery against strategic sybil attack in crowdsourcing,'' in \emph{Proceedings of the 26th ACM SIGKDD International Conference on Knowledge Discovery \& Data Mining}, 2020, pp. 95--104.

\bibitem{li2022harnessing}
F.~Li, J.~Zhao, D.~Yu, X.~Cheng, and W.~Lv, ``Harnessing context for budget-limited crowdsensing with massive uncertain workers,'' \emph{IEEE/ACM Transactions on Networking}, vol.~30, no.~5, pp. 2231--2245, 2022.

\bibitem{yuan2020distributed}
Y.~Yuan, F.~Li, D.~Yu, J.~Zhao, J.~Yu, and X.~Cheng, ``Distributed social learning with imperfect information,'' \emph{IEEE Transactions on Network Science and Engineering}, vol.~8, no.~2, pp. 841--852, 2020.

\bibitem{battaggion2022bright}
M.~R. Battaggion and G.~Karako{\c{c}}, ``On the bright side of correlation information transmission,'' \emph{Available at SSRN 4239807}, 2022.

\bibitem{crastr1982}
V.~P. Crawford and J.~Sobel, ``Strategic information transmission,'' \emph{Econometrica: Journal of the Econometric Society}, pp. 1431--1451, 1982.

\bibitem{karakocc2021cheap}
G.~Karako{\c{c}}, ``Cheap talk with multiple experts and uncertain biases,'' \emph{The BE Journal of Theoretical Economics}, 2021.

\bibitem{McGChe2013}
A.~McGee and H.~Yang, ``Cheap talk with two senders and complementary information,'' \emph{Games and Economic Behavior}, vol.~79, pp. 181--191, 2013.

\bibitem{ShiCoo2017}
S.~E. Lu, ``Coordination-free equilibria in cheap talk games,'' \emph{Journal of Economic Theory}, vol. 168, pp. 177--208, 2017.

\bibitem{bhattacharya2018optimality}
S.~Bhattacharya, M.~Goltsman, and A.~Mukherjee, ``On the optimality of diverse expert panels in persuasion games,'' \emph{Games and Economic Behavior}, vol. 107, pp. 345--363, 2018.

\bibitem{boleslavsky2016evolving}
R.~Boleslavsky and T.~R. Lewis, ``Evolving influence: Mitigating extreme conflicts of interest in advisory relationships,'' \emph{Games and economic behavior}, vol.~98, pp. 110--134, 2016.

\bibitem{chakraborty2010persuasion}
A.~Chakraborty and R.~Harbaugh, ``Persuasion by cheap talk,'' \emph{American Economic Review}, vol. 100, no.~5, pp. 2361--82, 2010.

\bibitem{li2019recommending}
Y.~Li, C.~Courcoubetis, and L.~Duan, ``Recommending paths: Follow or not follow?'' in \emph{IEEE INFOCOM 2019-IEEE Conference on Computer Communications}.\hskip 1em plus 0.5em minus 0.4em\relax IEEE, 2019, pp. 928--936.

\bibitem{amin2018evaluating}
M.~Amin-Naseri, P.~Chakraborty, A.~Sharma, S.~B. Gilbert, and M.~Hong, ``Evaluating the reliability, coverage, and added value of crowdsourced traffic incident reports from waze,'' \emph{Transportation research record}, vol. 2672, no.~43, pp. 34--43, 2018.

\bibitem{hao2023saving}
\BIBentryALTinterwordspacing
S.~Hao and L.~Duan. (2023) To save mobile crowdsourcing from cheap-talk: A game theoretic learning approach. Techincal Report. [Online]. Available: \url{https://arxiv.org/abs/2306.06791}
\BIBentrySTDinterwordspacing

\bibitem{2020Comprehending}
F.~Hu, ``Comprehending customer satisfaction with hotels : Data analysis of consumer-generated reviews,'' \emph{International Journal of Contemporary Hospitality Management}, 2020.

\bibitem{emons2019strategic}
W.~Emons and C.~Fluet, ``Strategic communication with reporting costs,'' \emph{Theory and Decision}, vol.~87, pp. 341--363, 2019.

\bibitem{tripkovic2021cluster}
S.~Tripkovic, P.~Svoboda, V.~Raida, and M.~Rupp, ``Cluster density in crowdsourced mobile network measurements,'' in \emph{2021 IEEE 93rd Vehicular Technology Conference (VTC2021-Spring)}.\hskip 1em plus 0.5em minus 0.4em\relax IEEE, 2021, pp. 1--7.

\bibitem{amini2013crowdlearner}
S.~Amini and Y.~Li, ``Crowdlearner: rapidly creating mobile recognizers using crowdsourcing,'' in \emph{Proceedings of the 26th annual ACM symposium on User interface software and technology}, 2013, pp. 163--172.

\bibitem{tao2018domain}
D.~Tao, J.~Cheng, Z.~Yu, K.~Yue, and L.~Wang, ``Domain-weighted majority voting for crowdsourcing,'' \emph{IEEE transactions on neural networks and learning systems}, vol.~30, no.~1, pp. 163--174, 2018.

\bibitem{xu2018reward}
J.~Xu, S.~Wang, N.~Zhang, F.~Yang, and X.~Shen, ``Reward or penalty: Aligning incentives of stakeholders in crowdsourcing,'' \emph{IEEE Transactions on Mobile Computing}, vol.~18, no.~4, pp. 974--985, 2018.

\bibitem{9364734}
C.~Huang, H.~Yu, J.~Huang, and R.~A. Berry, ``Eliciting information from heterogeneous mobile crowdsourced workers without verification,'' \emph{IEEE Transactions on Mobile Computing}, vol.~21, no.~10, pp. 3551--3564, 2022.

\bibitem{tullock1959problems}
G.~Tullock, ``Problems of majority voting,'' \emph{Journal of political economy}, vol.~67, no.~6, pp. 571--579, 1959.

\bibitem{saunders2006majority}
R.~S. Saunders, ``Why majority voting in director elections is a bad idea,'' \emph{Va. L. \& Bus. Rev.}, vol.~1, p. 107, 2006.

\bibitem{M1}
\BIBentryALTinterwordspacing
J.~Lamour. (2023, Feb) Montreal’s no. 1 restaurant on tripadvisor doesn’t actually exist. Today. [Online]. Available: \url{https://www.today.com/food/restaurants/tripadvisor-montreal-fake-restaurant-rcna69008}
\BIBentrySTDinterwordspacing

\bibitem{farrell1996cheap}
J.~Farrell and M.~Rabin, ``Cheap talk,'' \emph{Journal of Economic perspectives}, vol.~10, no.~3, pp. 103--118, 1996.

\bibitem{dipalantino2009crowdsourcing}
D.~DiPalantino and M.~Vojnovic, ``Crowdsourcing and all-pay auctions,'' in \emph{Proceedings of the 10th ACM conference on Electronic commerce}, 2009, pp. 119--128.

\bibitem{ghosh2012crowdsourcing}
A.~Ghosh and P.~McAfee, ``Crowdsourcing with endogenous entry,'' in \emph{Proceedings of the 21st international conference on World Wide Web}, 2012, pp. 999--1008.

\bibitem{ye2020federated}
D.~Ye, R.~Yu, M.~Pan, and Z.~Han, ``Federated learning in vehicular edge computing: A selective model aggregation approach,'' \emph{IEEE Access}, vol.~8, pp. 23\,920--23\,935, 2020.

\bibitem{li2021talk}
L.~Li, D.~Shi, R.~Hou, H.~Li, M.~Pan, and Z.~Han, ``To talk or to work: Flexible communication compression for energy efficient federated learning over heterogeneous mobile edge devices,'' in \emph{IEEE INFOCOM 2021-IEEE Conference on Computer Communications}.\hskip 1em plus 0.5em minus 0.4em\relax IEEE, 2021, pp. 1--10.

\bibitem{gibbons1992game}
R.~S. Gibbons, \emph{Game theory for applied economists}.\hskip 1em plus 0.5em minus 0.4em\relax Princeton University Press, 1992.

\bibitem{hao23sav}
S.~Hao and L.~Duan, ``To save crowdsourcing from cheap-talk: Strategic learning from biased users,'' in \emph{Proc. WiOpt}, 2023.

\end{thebibliography}
  % \vspace{-25pt}
 
\begin{IEEEbiography}[{\includegraphics[width=1in,height=1.25in,clip,keepaspectratio]{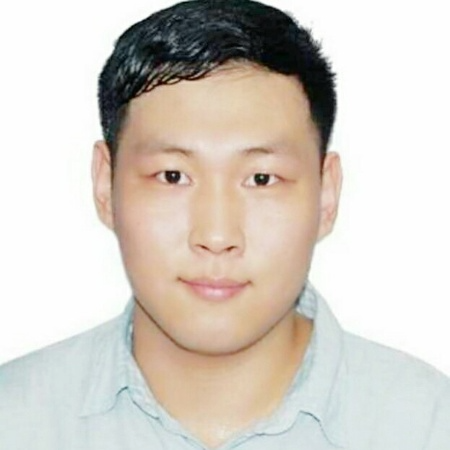}}]{Shugang Hao} (M'22)
 received the Ph.D. degree from Singapore University of Technology and Design (SUTD) in 2022. He is a postdoctoral research fellow at SUTD from Sep. 2022.  
 % He received B.Eng. degree in Information Engineering from South China University of Technology in 2017. 
His research interests are game theory and mechanism design. He received the student travel grant for attending ACM MobiHoc 2019 Symposium. He served as the local arrangement chair of IEEE WiOpt 2023. 
% In addition, he also received student travel grant for attending ACM MobiHoc 2019 Symposium.
\end{IEEEbiography}

 % \vspace{-35pt}

\begin{IEEEbiography}[{\includegraphics[width=1in,height=1.25in,clip,keepaspectratio]{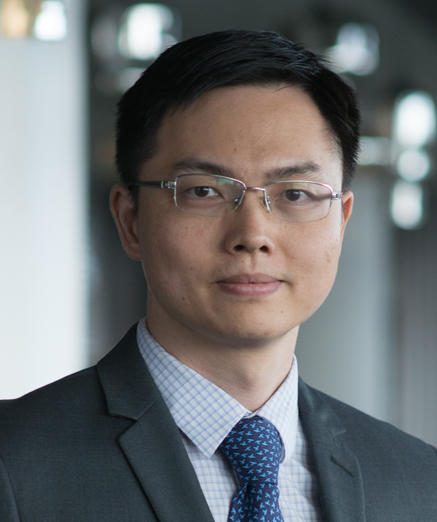}}]{Lingjie Duan}
(S'09-M'12-SM'17) received the
Ph.D. degree from The Chinese University of
Hong Kong in 2012. He is an Associate Professor and Associate Head of Pillar
of Engineering Systems and Design with the
Singapore University of Technology and Design
(SUTD). In 2011, he was a Visiting Scholar at University
of California at Berkeley, Berkeley, CA,
USA. His research interests include network economics
and game theory, cognitive communications,
and cooperative networking. He is an Associate
Editor of IEEE Transactions on Mobile Computing. He was an
Editor of IEEE Transactions on Wireless Communications and IEEE Communications Surveys and Tutorials. He also served
as a Guest Editor of the IEEE Journal on Selected Areas in Communications Special Issue on Human-in-the-Loop Mobile Networks, as well
as IEEE Wireless Communications Magazine.  He received the SUTD Excellence
in Research Award in 2016 and the 10th IEEE ComSoc Asia-Pacific
Outstanding Young Researcher Award in 2015. He served as the general chair of IEEE WiOpt 2023.
\end{IEEEbiography}

  % \begin{comment}

\newpage

 ~\newpage

% \onecolumn
\appendices

\section{Proof of Lemma~\ref{lemma-bm-2}}
$\Bar{a}_1$ in \eqref{a_1} can be easily obtained due to $c(a, \theta)$'s convexity in $a$. $\Bar{L}_1$ in \eqref{u-r-b-2} can be thus obtained by substituting $\Bar{a}_1$ in \eqref{a_1} to \eqref{u-r} and \eqref{POA}.

To prove that $\bar{L}_1$ in \eqref{u-r-b-2} does not monotonically decrease with $N$, we give a counter-example by taking $p_H$=0.5 and $N$=1, 2, 3. We can then calculate $\bar{L}_1$ in \eqref{u-r-b-2} under different $N$ as follows:
\begin{align*}
    &\bar{L}_1(N=1) = \frac{1}{2}(\mu_H-\mu_L)^2, \\
    &\bar{L}_1(N=2) = \frac{3}{8}(\mu_H-\mu_L)^2, \\
    &\bar{L}_1(N=3) = \frac{1}{2}(\mu_H-\mu_L)^2,
\end{align*}
We find that $\bar{L}_1(N=1) > \bar{L}_1(N=2)$ but $\bar{L}_1(N=2) < \bar{L}_1(N=3)$, which implies that $\bar{L}_1$ in \eqref{u-r-b-2} does not monotonically decrease with $N$.

To prove $\Bar{L}_1 > p_H(1-p_H)(\mu_H\!-\!\mu_L)^2$, we have 
\begin{align*}
    &\bar{L}_1 > (\mu_H\!-\!\mu_L)^2\bigg(p_H(1-p_H)\!\!\sum_{k=0}^{\lceil\frac{N}{2}\!-\!1\rceil}C_N^k\frac{1}{2^N} \\
    &+\!p_H(1\!-\!p_H) \!\!\!\sum_{l=\lfloor\frac{N}{2}\!+\!1\rfloor}^{N} \!\!\!C_N^l\frac{1}{2^N} +\text{$\mathbb{1}_{\lfloor\frac{N}{2}\rfloor=\frac{N}{2}}$}\cdot p_H(1\!-\!p_H)C_N^\frac{N}{2}\frac{1}{2^N}\bigg) \\
    &=p_H(1-p_H)(\mu_H\!-\!\mu_L)^2
\end{align*}
due to $p_H \in (0, 1)$. We then finsh the proof.

\section{Proof of Lemma~\ref{L11}}

Before we formally prove the lemma, let us first list the user's remaining 10 possible strategy combinations (except for strategy combinations 1-6 mentioned in the lemma) as follows:
\begin{itemize}
    \item users' strategy combination 7: honest messaging strategy for both $b_i=-$ and $b_i=+$,
    \item users' strategy combination 8: blindly high-PDF type messaging strategy for both $b_i=-$ and $b_i=+$,
    \item users' strategy combination 9: blindly high-PDF type messaging strategy for $b_i=-$ and blindly low-PDF type messaging strategy for $b_i=+$,
    \item users' strategy combination 10: blindly low-PDF type messaging strategy for both $b_i=-$ and $b_i=+$,
    \item users' strategy combination 11: blindly low-PDF type messaging strategy for $b_i=-$ and blindly high-PDF type messaging strategy for $b_i=+$,
    \item users' strategy combination 12: blindly high-PDF type messaging strategy for $b_i=-$ and honest messaging strategy for $b_i=+$,
    \item users' strategy combination 13: blindly low-PDF type messaging strategy for $b_i=-$ and reversed messaging strategy for $b_i=+$,
    \item users' strategy combination 14: honest messaging strategy for $b_i=-$ and blindly low-PDF type messaging strategy for $b_i=+$,
    \item users' strategy combination 15: reversed messaging strategy for $b_i=-$ and blindly high-PDF type messaging strategy for $b_i=+$,
    \item users' strategy combination 16: reversed messaging strategy for both $b_i=-$ and $b_i=+$.
\end{itemize}
Denote set $I_j$=$\{i|$$m_i$=$\phi_j(\theta)$,1$\leq$$i$$\leq$$N\}$ to include all users for messaging $\phi_j(\theta)$ with $j$$\in$$\{L, H\}$.
 Similar to analysis on users' strategy combination 1 in \eqref{ul'-2}-\eqref{u-1}, we obtain that the platform's best-response actions in \eqref{a} to users' strategy combination 7-16 as follows:
\begin{align}
    &a_7^*(m) = \begin{cases}
        \mu_H, &\text{if} \ |I_H| = N, \\
        \mu_L, &\text{if} \ |I_L| = N,
    \end{cases} \label{a-7} \\
% \end{align*}
% \begin{align*}
    &a_8^*(m) = a_9^*(m) = a_{10}^*(m) = a_{11}^*(m) \nonumber \\ 
    &= p_H\mu_H+(1-p_H)\mu_L,\label{a-891011} \\
% \end{align*}
% \begin{align*}
    &a_{12}^*(m) = \begin{cases}
        \frac{p_H\mu_H+(1-p_H)(1-q_+)^N\mu_L}{p_H+(1-p_H)(1-q_+)^N}, &\text{if} \ |I_H| = N, \\
        \mu_L, &\text{if} \ |I_L| \geq 1,
    \end{cases} \label{a-12}\\
% \end{align*}
% \begin{align*}
    &a_{13}^*(m) = \begin{cases}
       \mu_L, &\text{if} \ |I_H| \geq 1, \\
       \frac{p_H\mu_H+(1-p_H)(1-q_+)^N\mu_L}{p_H+(1-p_H)(1-q_+)^N}, &\text{if} \ |I_L| = N,
    \end{cases} \label{a-13} \\
% \end{align*}
% \begin{align*}
    &a_{14}^*(m) = \begin{cases}
        \mu_H, &\text{if} \ |I_H| \geq 1, \\
        \frac{p_H q_+^N\mu_H+(1-p_H)\mu_L}{p_H q_+^N+(1-p_H)}, &\text{if} \ |I_L| = N,
    \end{cases} \label{a-14} \\
% \end{align*}
% \begin{align*}
    &a_{15}^*(m) = \begin{cases}
        \frac{p_H q_+^N\mu_H+(1-p_H)\mu_L}{p_H q_+^N+(1-p_H)}, &\text{if} \ |I_H| = N, \\
        \mu_H, &\text{if} \ |I_L| \geq 1.
    \end{cases} \label{a-15} \\
% \end{align*}
% \begin{align*}
    &a_{16}^*(m) = \begin{cases}
        \mu_L, &\text{if} \ |I_H| = N, \\
        \mu_H, &\text{if} \ |I_L| = N.
    \end{cases} \label{a-16}
\end{align}
According to the platform's actions in \eqref{a-2}-\eqref{a-6} and \eqref{a-7}-\eqref{a-16}, we are able to give the user's expected utility in \eqref{u-s} under strategy combinations 2-16.

Next, we will show that users' strategy combination 7 is  not stable to occur at the PBE, respectively. To show each user $i$ with bias $b_i=-$ always deviates from strategy combination 7 to 1, 
 % according to \eqref{u-789101116} and \eqref{u-1}, 
we have 
\begin{align*}
    &\bar{u}_7(b_i=-) - \bar{u}_1(b_i=-) \\
    =& -\frac{(1-p_H)p_Hq_+(1-q_+)^{N-1}(\mu_H-\mu_L)}{1 - p_H (1-(1-q_+)^N) } \leq 0
\end{align*}
due to $p_H \in [0, 1]$, $q_+ \in [0, 1]$ and $\mu_H \geq \mu_L$. 
By following similar steps, we can also show user's strategy combinations 8-16 is not stable to occur at the PBE, respectively.  We then finish the proof.

\section{Proof of Proposition~\ref{prop-11}}

In the arbitrary $N$-user case, according to the platform's actions in \eqref{a2-1} and \eqref{a2-3}-\eqref{a2-6}, we are able to give each user $i$'s expected utility in \eqref{u-s} under strategy combinations 1-16
as \eqref{u2-12}-\eqref{u2-1415}.
\begin{table*}
\begin{align}
    &\bar{u}_1(b_i) = \bar{u}_3(b_i) = \label{u2-12}
    \begin{cases}
    \frac{\left(-p_H (1-q_+)^{N-1}+p_H-1\right) \left(p_H \mu _H
   (1-q_+)^N+(1-p_H) \mu _L\right)}{p_H (1-q_+)^N-p_H+1}-p_H \mu _H
   \left(1-(1-q_+)^{N-1}\right), &\text{if} \ b_i = -,  \\
    \frac{\left(1-p_H\right) \left(\left(1-p_H\right) \mu _L+\mu _H
   p_H \left(1-q_+\right){}^N\right)}{p_H
   \left(1-q_+\right){}^N-p_H+1}+\mu _H p_H, &\text{if} \ b_i \!=\! +. 
    \end{cases} \\
    &\bar{u}_2(b_i) = \bar{u}_4(b_i) = \label{u2-34}
    \begin{cases}
    -p_H \frac{p_H\mu_H + (1-p_H)q_+^N\mu_L}{p_H + (1-p_H)q_+^N} - (1-p_H)\mu_L, &\text{if} \ b_i = -,  \\
    \frac{\left(1-p_H\right) q_+^{N-1} \left(\left(1-p_H\right) \mu
   _L q_+^N+\mu _H p_H\right)}{\left(1-p_H\right)
   q_+^N+p_H}+\left(1-p_H\right) \mu _L
   \left(1-q_+^{N-1}\right)+\frac{p_H \left(\left(1-p_H\right)
   \mu _L q_+^N+\mu _H p_H\right)}{\left(1-p_H\right) q_+^N+p_H}, &\text{if} \ b_i \!=\! +. 
    \end{cases} \\
    &\bar{u}_5(b_i) = \begin{cases}
        -p_H \sum_{k=0}^{N-1} C_{N-1}^k q_+^k (1-q_+)^{N-1-k} \frac{p_H q_+^k (1-q_+)^{N-k} \mu_H+(1-p_H)(1-q_+)^k q_+^{N-k}\mu_L}{p_H q_+^k (1-q_+)^{N-k}+(1-p_H)(1-q_+)^kq_+^{N-k}} &\\
        -(1-p_H) \sum_{k=1}^{N} C_{N-1}^{k-1}(1-q_+)^{k-1}q_+^{N-k} \frac{p_H q_+^k (1-q_+)^{N-k} \mu_H+(1-p_H)(1-q_+)^k q_+^{N-k}\mu_L}{p_H q_+^k (1-q_+)^{N-k}+(1-p_H)(1-q_+)^k q_+^{N-k}}, &\text{if} \ b_i=-, \\
        p_H \sum_{k=1}^{N} C_{N-1}^{k-1} q_+^{k-1} (1-q_+)^{N-k} \frac{p_H q_+^k (1-q_+)^{N-k} \mu_H+(1-p_H)(1-q_+)^k q_+^{N-k}\mu_L}{p_H q_+^k (1-q_+)^{N-k}+(1-p_H)(1-q_+)^kq_+^{N-k}} &\\
        +(1-p_H) \sum_{k=0}^{N-1} C_{N-1}^{k}(1-q_+)^{k}q_+^{N-k-1} \frac{p_H q_+^k (1-q_+)^{N-k} \mu_H+(1-p_H)(1-q_+)^k q_+^{N-k}\mu_L}{p_H q_+^k (1-q_+)^{N-k}+(1-p_H)(1-q_+)^k q_+^{N-k}}, &\text{if} \ b_i=+,
    \end{cases} \label{u2-5} \\
    &\bar{u}_6(b_i) = \begin{cases}
         -p_H \sum_{k=1}^{N} C_{N-1}^{k-1} (1-q_+)^{k-1} q_+^{N-k} \frac{p_H (1-q_+)^k q_+^{N-k} \mu_H+(1-p_H)q_+^k (1-q_+)^{N-k}\mu_L}{p_H (1-q_+)^k q_+^{N-k} +(1-p_H)q_+^k (1-q_+)^{N-k}} &\\
        -(1-p_H) \sum_{k=0}^{N-1} C_{N-1}^{k}q_+^{k}(1-q_+)^{N-k-1} \frac{p_H (1-q_+)^k q_+^{N-k} \mu_H+(1-p_H)q_+^k (1-q_+)^{N-k}\mu_L}{p_H (1-q_+)^k q_+^{N-k}+(1-p_H)q_+^k (1-q_+)^{N-k}}, &\text{if} \ b_i=-, \\
        p_H \sum_{k=0}^{N-1} C_{N-1}^k (1-q_+)^k q_+^{N-1-k} \frac{p_H (1-q_+)^k q_+^{N-k} \mu_H+(1-p_H)q_+^k (1-q_+)^{N-k}\mu_L}{p_H (1-q_+)^k q_+^{N-k} +(1-p_H)q_+^k (1-q_+)^{N-k}} &\\
        +(1-p_H) \sum_{k=1}^{N} C_{N-1}^{k-1}q_+^{k-1}(1-q_+)^{N-k} \frac{p_H (1-q_+)^k q_+^{N-k} \mu_H+(1-p_H)q_+^k (1-q_+)^{N-k}\mu_L}{p_H (1-q_+)^k q_+^{N-k} +(1-p_H)q_+^k (1-q_+)^{N-k}}, &\text{if} \ b_i=+,
    \end{cases} \label{u2-6} \\
    &\bar{u}_7(b_i) = \bar{u}_8(b_i) = \bar{u}_9(b_i) = \bar{u}_{10}(b_i) = \bar{u}_{11}(b_i) = \bar{u}_{16}(b_i) = \begin{cases}
        -p_H \mu_H - (1-p_H)\mu_L, &\text{if} \ b_i=-, \\
        p_H \mu_H + (1-p_H)\mu_L, &\text{if} \ b_i=+,
    \end{cases} \label{u2-789101116} \\
    &\bar{u}_{12}(b_i) = \bar{u}_{13}(b_i) =  \begin{cases}
     -(p_H+(1-p_H)(1-q_+)^{N-1})\frac{p_H\mu_H+(1-p_H)(1-q_+)^N\mu_L}{p_H+(1-p_H)(1-q_+)^N}-(1-p_H)(1-(1-q_+)^{N-1})\mu_L, &\text{if} \ b_i=-, \\
     p_H\frac{p_H\mu_H+(1-p_H)(1-q_+)^N\mu_L}{p_H+(1-p_H)(1-q_+)^N}+(1-p_H)\mu_L, &\text{if} \ b_i=+,
    \end{cases} \label{u2-1213} \\
    &\bar{u}_{14}(b_i) = \bar{u}_{15}(b_i) =  \begin{cases}
     -p_H \mu_H - (1-p_H)\frac{p_H q_+^N \mu_H + (1-p_H)\mu_L}{p_H q_+^N  + (1-p_H)}, &\text{if} \ b_i=-, \\
    (p_H q_+^{N-1}+1-p_H) \frac{p_H q_+^N \mu_H + (1-p_H)\mu_L}{p_H q_+^N  + (1-p_H)} + p_H(1-q_+^{N-1})\mu_H, &\text{if} \ b_i=+.
    \end{cases} \label{u2-1415} 
\end{align}
\end{table*}

Lemma~\ref{L11} holds for arbitrary $N$ users and we thus know that only users' strategy combinations 1-6 may occur at the PBE. First, we will prove that each user $i$'s strategy combination 2 always occurs at the PBE (strategy combination 4's occurring at the PBE can be proved by similar method and we thus omit here). Similar to what we have done for one-user case in Appendix~\ref{A-B}, the only possible deviation from strategy combination 2 is to strategy 6 with the user $i$'s bias $b_i=+$, which still never occurs due to
\begin{align*}
    \bar{u}_2(b_i\!=\!+) \!-\! \bar{u}_6(b_i\!=\!+) \!= \!\frac{(1-p_H)p_H(\mu_H-\mu_L)}{1+(2^N-1)p_H} \!>\! 0, \ q_+ \!=\! \frac{1}{2}.
\end{align*}
Thus, each user $i$ with bias $b_i=-$ or $b_i=+$ never deviates from strategy combination 2 and it is always a PBE.

By using similar method as above, we can also prove that each user's strategy combination 1 always occurs at the PBE (strategy combination 3's occurring at the PBE can be proved by similar method and we thus omit here).

Finally, we have shown users with bias $b=+$ always deviate from strategy combination 6 to 2. Strategy combination 6 is thus never a PBE (strategy combination 5's occurring at the PBE can be proved by similar method and we thus omit here). We then finish the proof.

\section{Proof of Theorem~\ref{prop-12}}

The system loss $\bar{L}_{Bayesian}$ in \eqref{N} can be easily obtained according to Theorem~\ref{prop-11} and \eqref{u-r}. Besides, one can easily find that $\bar{L}_{Bayesian} \to 0$ as $N \to \infty$. In the following we will show 
loss difference $\bar{L}_2-\bar{L}_{Bayesian}$ in the regime of $p_H \leq \frac{1}{2}$, which can be obtained similarly in the other regime of $p_H > \frac{1}{2}$.  To show $\bar{L}_{Bayesian}$ is less than $\bar{L}_1$ in \eqref{u-r-b-2}, we have
\begin{align*}
    \bar{L}_{Bayesian} < p_H(1-p_H)(\mu_H-\mu_L)^2
\end{align*}
and
\begin{align*}
    \bar{L}_1 &\geq \max\{p_H(\mu_H-\mu_L)^2, (1-p_H)(\mu_H-\mu_L)^2\} \\
    &> p_H(1-p_H)(\mu_H-\mu_L)^2.
\end{align*}
We thus have $\bar{L}_{Bayesian} < \bar{L}_1$ and 
\begin{align*}
    \Delta L_1 &= \bar{L}_1 - \bar{L}_{Bayesian} \\
    &> \bar{L}_2 - \bar{L}_{Bayesian} = \frac{p_H(1-p_H)^2(2^N\!-\!1)(\mu_H-\mu_L)^2}{p_H+(1-p_H)2^N},
\end{align*}
the lower-bound on $\frac{p_H(1-p_H)^2(2^N\!-\!1)(\mu_H-\mu_L)^2}{p_H+(1-p_H)2^N}$ is obvious to be 0 as $p_H$=0 or 1. The upper-bound can be obtained by checking derivative on $p_H$ and applying first-order condition, which returns 
\begin{align*}
    p_H^* = \frac{3\ 2^{N-2}}{2^N-1}-\frac{1}{4} \sqrt{\frac{2^{2N}+2^{N+3}}{\left(2^N-1\right)^2}}. 
\end{align*}
After substituting the above $p_H^*$ to $\frac{p_H(1-p_H)^2(2^N\!-\!1)(\mu_H-\mu_L)^2}{p_H+(1-p_H)2^N}$, we obtain the upper bound. After checking its derivatives on $N$, we obtain that the upper bound increases with $N$.  We then finish the proof.

\begin{table*}
\begin{align}
    &L = -p_H \sum_{k=1}^T  \int_{\Theta_{k-1}}  ( a_H^2(h_{k-1}) - 2\mu_H a_H(h_{k-1}) + \sigma^2 + \mu_H^2 ) \phi_H(h_{k-1}) dh_{k-1} + \lambda_L \sum_{k=1}^T  \int_{\Theta_{k-1}}  ( a_L(h_{k-1}) - a_H(h_{k-1}) ) \phi_L(h_{k-1}) dh_{k-1} \nonumber \\
    &- (1-p_H) \sum_{k=1}^T  \int_{\Theta_{k-1}}  ( a_L^2(h_{k-1}) - 2\mu_L a_L(h_{k-1}) + \sigma^2 + \mu_L^2) \phi_L(h_{k-1}) dh_{k-1} + \lambda_H  \sum_{k=1}^T  \int_{\Theta_{k-1}}  ( a_L(h_{k-1}) - a_H(h_{k-1}) ) \phi_H(h_{k-1}) dh_{k-1}. \label{A-L2}
\end{align}
\begin{align}
    &a_L(h_{k\!-\!1}) \!=\! \frac{\phi_H(h_{k\!-\!1}) \lambda _H+2 \phi_L(h_{k\!-\!1}) \left(1-p_H\right) \mu _L+\phi_L
   (h_{k-1}) \lambda _L}{2 \phi_L(h_{k\!-\!1}) \left(1-p_H\right)}, \label{aL-1} \\
   &a_H(h_{k\!-\!1}) = 
   \frac{\phi_H(h_{k\!-\!1}) \lambda_H-2 \phi_H(h_{k\!-\!1}) \mu_H\!+\!2 \phi_H(h_{k\!-\!1}) \mu_H
   \left(1-p_H\right)\!+\!\phi_L(h_{k\!-\!1}) \lambda_L}{2 \phi_H(h_{k\!-\!1}) p_H}. \label{aH-1}
\end{align} 
\begin{align}
     \lambda_L \!=\! \frac{2 (1\!-\!p_H) p_H^2 (-\mu_H s_{\beta  }
   s_{\delta }\!+\!\mu _H s_{\delta }^2\!+\!\mu _L s_{\beta }
   s_{\delta }\!-\!\mu _L s_{\delta
   }^2)}{p_H^2 (s_{\alpha  }\!-\!s_{\delta }) (s_{\beta}\!-\!s_{\delta })\!-\!p_H (s_{\alpha }
   (s_{\beta }\!-\!2 s_{\delta })\!+\!s_{\delta
   }^2)\!+\!s_{\delta } (s_{\delta }\!-\!s_{\alpha })} > 0, \ \lambda_H \!=\! \frac{2 \left(1\!-\!p_H\right)^2 p_H (-\mu _H s_{\alpha} s_{\delta }\!+\!\mu _H s_{\delta }^2\!+\!\mu _L s_{\alpha
   \delta } s_{\delta }\!-\!\mu _L s_{\delta
   }^2)}{p_H^2 (s_{\alpha }\!-\!s_{\delta }) (s_{\beta}-s_{\delta })\!-\!p_H (s_{\alpha }
   (s_{\beta }\!-\!2 s_{\delta })\!+\!s_{\delta
   }^2)\!+\!s_{\delta } (s_{\delta }\!-\!s_{\alpha })} > 0, \label{A-L3}
\end{align}
\end{table*}

\section{Proof of Lemma~\ref{3-L}}
     To ensure the positively-biased user type's no deviation from messaging of true PDF type to the other, the platform needs to satisfy the following incentive-compatibility (IC) constraints according to \eqref{u_s_} regarding low- and high-PDF type realizations, respectively: 
\begin{align}
    &\int_{\theta} (a_L\!-\!a_H) \phi_L(\theta) d\theta \geq 0, \int_{\theta} (a_H\!-\!a_L) \phi_H(\theta) d\theta \geq 0. \label{IC-1-1} 
\end{align}
Similarly, to ensure the negatively-biased user type's no deviation from messaging of true PDF type to the other, the platform needs to satisfy the following IC constraints according to \eqref{u_s_} regarding low- and high-PDF type realizations, respectively:
\begin{align}
    &\int_{\theta} (a_H\!-\!a_L) \phi_L(\theta) d\theta \geq 0, \int_{\theta} (a_L\!-\!a_H) \phi_H(\theta) d\theta \geq 0. \label{IC-3-1}
\end{align}

The platform's objective is to minimize its expected cost. Therefore, the platform aims to find best $a_L$ and $a_H$ by solving the following minimization problem:
\begin{align}\label{Problem'}
    \min_{a_L,a_H} &\int_\theta \sum_{j\in\{H, L\} }Pr(\phi(\theta) = \phi_j(\theta)) u_R(a_j, \theta) \phi_j(\theta) d\theta \nonumber \\
s.t. \ &\eqref{IC-1-1}, \eqref{IC-3-1}.
\end{align}
The only solution to \eqref{IC-1-1} and \eqref{IC-3-1} is $a_L = a_H$, which means the mechanism degenerates to the benchmark 2. We then finish the proof.

\section{Proof of Theorem~\ref{prop-14}}

We first reduce the original problem \eqref{Problem} to the following: 
\begin{align}\label{P'}
    \max_{\{a_L(h_{k-1})\}_{k=1}^N, \{a_H(h_{k-1})\}_{k=1}^N} &-\eqref{u-r-md} \nonumber \\
s.t. \ &\eqref{IC-1}, \eqref{IC-3}.
\end{align}
Later we will show that solutions to \eqref{P'} also satisfy the two relaxed constraints \eqref{IC-2} and \eqref{IC-4}. 

Note that in the relaxed problem \eqref{P'}, the objective $\bar{c}^T$ in \eqref{u-r-md} is concave with each $a_L(h_{k-1})$ and $a_H(h_{k-1})$, respectively, and the constraints \eqref{IC-1}-\eqref{IC-3} are linear in each $a_L(h_{k-1})$ and $a_H(h_{k-1})$, respectively. The Largrangian for problem \eqref{P'} is then given as in \eqref{A-L2}.
Stationary conditions are 
\begin{align*}
    \frac{\partial L}{\partial a_L(h_{k-1})} = 0, \ \frac{\partial L}{\partial a_H(h_{k-1})} = 0, \ k \in \{1, \cdots, N\},
\end{align*}
which have the solutions as \eqref{aL-1} and \eqref{aH-1}.

At the optimum, we have constraints \eqref{IC-1} and \eqref{IC-3} to be equalities. After substituting $a_L(h_{k-1})$ in \eqref{aL-1} and $a_H(h_{k-1})$ in \eqref{aH-1} in equal \eqref{IC-1} and \eqref{IC-3}, respectively, we have $\lambda_L$ and $\lambda_H$ as the unique solutions as in \eqref{A-L3},
where
\begin{align*}
    % &\Lambda(h_K) = \frac{f_H(h_K)}{f_L(h_K)}, \\
    &s_{\delta } \!=\! T,\ s_{\alpha } \!= \!\sum_{k=1}^T(\alpha)^{k\!-\!1}, \ s_{\beta  } \!= \!\sum_{k=1}^T(\beta)^{k\!-\!1}, \\
    &\alpha = \int_\Theta \frac{\phi_L^2(\theta)}{\phi_H(\theta)}d\theta, \ \beta = \int_\Theta \frac{\phi_H^2(\theta)}{\phi_L(\theta)}d\theta.
\end{align*}
By substituting $\lambda_L$, $\lambda_H$ in \eqref{aL-1} and \eqref{aH-1}, respectively, we have
\begin{align*}
    &a_L^*(h_{k-1}) = \mu_L + w(h_{k-1})\left(\mu _H-\mu
   _L\right) > \mu_L, \\
   &a_H^*(h_{k-1}) = \mu_H - w(h_{k-1}) \frac{1-p_H}{\Lambda(h_{k-1}) p_H}\left(\mu _H-\mu
   _L\right) < \mu_H, 
\end{align*}
where
\begin{align*}
    w(h_{k-1}) = \frac{\Lambda(h_{k-1}) \lambda_H + \lambda_L}{2(\mu _H-\mu
   _L)(1-p_H)} > 0.
\end{align*}
Note that when \eqref{IC-1} and \eqref{IC-3} become equalities, we have \eqref{IC-2} and \eqref{IC-4} also become equalities and thus satisfied. Eventually, the above ($a_L^*(h_{k-1})$, $a_H^*(h_{k-1})$) are also solutions to the original problem \eqref{Problem}. 

The system loss $\hat{L}_{evolving}$ of the time-evolving commitment mechanism in \eqref{u-r-m} can be easily obtained according to $a_L^*(h_{k-1})$ and $a_H^*(h_{k-1})$ in Theorem~\ref{prop-14} and \eqref{u-r}. Besides, one can easily show that $\hat{L}_{evolving}$ decreases with $\alpha$ and $\beta$, respectively, by checking corresponding derivatives. To prove $\hat{L}_{evolving} < p_H(1-p_H)(\mu_H-\mu_L)^2$, note that $\hat{L}_{evolving}$ decreases with $\alpha$ and $\alpha > 1$ by definition, we then have $\hat{L}_{evolving} < \hat{L}_{evolving}(\alpha = 1) = p_H(1-p_H)(\mu_H-\mu_L)^2$. 

\section{Proof of Proposition~\ref{prop-1}}\label{A-B}

According to Lemma~\ref{L11}, we know that only user's strategy combinations 1-6 may occur at the PBE. Firstly, we will prove that user's strategy combination 2 always occurs at the PBE if $q_+ < \frac{3-\sqrt{5}}{2}$ (strategy combination 4's occurring at the PBE can be proved by similar method and we thus omit here). To show the user with bias $b=-$ never deviates from strategy combination 2 to 8 or 9, 
% according to \eqref{u-34} and \eqref{u-789101116}, 
we have 
\begin{align*}
   &\bar{u}_{2}(b=-) - \bar{u}_{8}(b=-) = \bar{u}_{2}(b=-) - \bar{u}_{9}(b=-) \\
   =& \frac{(1-p_H)p_H(1-q_+)(\mu_H-\mu_L)}{p_H(1-q_+)+q_+} \geq 0
\end{align*}
due to $p_H \in [0, 1]$, $q_+ \in [0, 1]$ and $\mu_H \geq \mu_L$. To show the user with bias $b=-$ never deviates from strategy combination 2 to 15 
% according to \eqref{u-34} and \eqref{u-1415}
, we have 
\begin{align*}
   &\bar{u}_{2}(b=-) - \bar{u}_{15}(b=-) \\
   =& \frac{(1-p_H)p_H(1+q_+)q_+(\mu_H-\mu_L)}{(1-p_H(1-q_+))(p_H(1-q_+)+q_+)} \geq 0
\end{align*}
due to $p_H \in [0, 1]$, $q_+ \in [0, 1]$ and $\mu_H \geq \mu_L$. To show the user with bias $b=+$ never deviates from strategy combination 2 to 7 
% according to \eqref{u-34} and \eqref{u-789101116}
, we have 
\begin{align*}
   &\bar{u}_{2}(b=+) - \bar{u}_{7}(b=+) \\
   =& \frac{(1-p_H)p_H(1-q_+)(\mu_H-\mu_L)}{p_H(1-q_+)+q_+} \geq 0
\end{align*}
due to $p_H \in [0, 1]$, $q_+ \in [0, 1]$ and $\mu_H \geq \mu_L$. To show the user with bias $b=+$ never deviates from strategy combination 2 to 14 
% according to \eqref{u-34} and \eqref{u-1415}
, we have 
\begin{align*}
   &\bar{u}_{2}(b=+) - \bar{u}_{14}(b=+) \\
   =& \frac{(1-p_H)p_H(1-q_+^2)(\mu_H-\mu_L)}{(p_H(1-q_+)+q_+)(1-p_H(1-q_+))} \geq 0
\end{align*}
due to $p_H \in [0, 1]$, $q_+ \in [0, 1]$ and $\mu_H \geq \mu_L$. To show the user with bias $b=+$ never deviates from strategy combination 2 to 6 
% according to \eqref{u-34} and \eqref{u-1415}
, we have 
\begin{align*}
   &\bar{u}_{2}(b=+) - \bar{u}_{6}(b=+) = -\frac{(q_+-1) (p_H-1) p_H (\mu
   _H-\mu _L)}{((q_+-1)
   p_H-q_+)}  \cdot \\
   & \frac{ ((2 q_+-1) p_H^2+q_+ (2
   q_+-1) p_H-q_+^2)}{ ((2 q_+-1) p_H-q_+)
   ((2 q_+-1) p_H-q_++1)}.
\end{align*}
Note that we have
\begin{align}\label{A1}
    &\bar{u}_{2}(b=+) \geq \bar{u}_{6}(b=+) 
     \\ 
    &\Longleftrightarrow \frac{(2 q_+-1) p_H^2+q_+ (2
   q_+-1) p_H-q_+^2}{((2 q_+-1) p_H-q_+)
   ((2 q_+-1) p_H-q_++1)} \geq 0 \nonumber
\end{align}
due to $ -\frac{(q_+-1) (p_H-1) p_H (\mu_H-\mu _L)}{((q_+-1)p_H-q_+)} \geq 0$
for all $p_H \in [0, 1]$, $q_+ \in [0, 1]$ and $\mu_H \geq \mu_L$. Besides, we have $((2 q_+-1) p_H-q_+)((2 q_+-1) p_H-q_++1) \leq 0$ is equivalent to $p_H \leq \frac{1-q_+}{1-2q_+}$ if $q_+ \leq \frac{1}{2}$ and $p_H \leq \frac{q_+}{2q_+-1}$ if $q_+ \geq \frac{1}{2}$, which are both true because $\frac{1-q_+}{1-2q_+} \geq 1$ for $q_+ \leq \frac{1}{2}$ and $\frac{q_+}{2q_+-1} \geq 1$ for $q_+ \geq \frac{1}{2}$. Thus, \eqref{A1} can be rewritten as
\begin{align}\label{A2}
    &\bar{u}_{2}(b=+) \geq \bar{u}_{6}(b=+) \nonumber \\
    \Longleftrightarrow& (2 q_+-1) p_H^2+q_+ (2
   q_+-1) p_H-q_+^2 \leq 0.
\end{align}
One can easily solve \eqref{A2} to be 
\begin{align}\label{A3}
  &\bar{u}_{2}(b=+) \geq \bar{u}_{6}(b=+) \nonumber \\
  \Longleftrightarrow &\bigg( q_+ \leq \frac{\sqrt{5}-1}{2} \bigg) \ \text{or} \ \nonumber \\
  &\!\!\!\!\!\! \bigg(  p_H \leq p_{1,L} = \frac{1}{2} \sqrt{\frac{2 q_+^3+3 q_+^2}{2 q_+-1}}-\frac{q_+}{2} \ \text{and} \  q_+ \geq \frac{\sqrt{5}-1}{2} \bigg).
\end{align}
Since $q_+ \leq \frac{3-\sqrt{5}}{2}$ is contained in $q_+ \leq \frac{\sqrt{5}-1}{2}$ in \eqref{A3}, we have $\bar{u}_{3}(b=+) \geq \bar{u}_{6}(b=+) $ and the user with bias $b=+$ never deviates from strategy combination 2 to 6. To summarize, the user with either bias $b=+$ or bias $b=-$ never deviates from strategy combination 2 in the regime of $q_+ \leq \frac{3-\sqrt{5}}{2}$, and strategy combination 2 is thus a PBE for all $p_H \in [0, 1]$ and $q_+ \leq \frac{3-\sqrt{5}}{2}$.

By following similar steps, we can prove that user's strategy combination 1 always occurs at the PBE if $q_+ < \frac{3-\sqrt{5}}{2}$ and $p_H \geq p_{1,H}$ (strategy combination 3's occurring at the PBE can be proved by similar method and we thus omit here), and user's strategy combination 6 never occurs at the PBE (strategy combination 5's never occuring at the PBE can be proved by similar method and we thus omit here). We then finish the proof.

\section{Proof of Proposition~\ref{prop-2}}

PBE is given in Table~\ref{t1:prop-2}.
According to Lemma~\ref{L11}, we know that only user's strategy combinations 1-6 may occur at the PBE. First, we will prove that user's strategy combination 2 always occurs at the PBE if $q_+ \in [\frac{3-\sqrt{5}}{2}, \frac{\sqrt{5}-1}{2}]$ (strategy combination 4's occurring at the PBE can be proved by similar method and we thus omit here). We have already shown in Appendix~\ref{A-B} that the only possible deviation from strategy combination 2 is to strategy 6 with the user's bias $b=+$, which still never occurs if $q_+ \in [\frac{3-\sqrt{5}}{2}, \frac{\sqrt{5}-1}{2}]$ according to \eqref{A3}. Thus, the user with bias $b=-$ or $b=+$ never deviates from strategy combination 2 and it is a PBE for $q_+ \in [\frac{3-\sqrt{5}}{2}, \frac{\sqrt{5}-1}{2}]$.

Next, we will prove that user's strategy combination 1 always occurs at the PBE if $q_+ \in [\frac{3-\sqrt{5}}{2}, \frac{\sqrt{5}-1}{2}]$ (strategy combination 3's occurring at the PBE can be proved by similar method and we thus omit here). We have already shown in Appendix~\ref{A-B} that the only possible deviation from strategy combination 1 is to strategy 5 with the user's bias $b=-$, which still never occurs if $q_+ \in [\frac{3-\sqrt{5}}{2}, \frac{\sqrt{5}-1}{2}]$ 
% according to \eqref{A6}
. Thus, the user with bias $b=-$ or $b=+$ never deviates from strategy combination 1 and it is a PBE for $q_+ \in [\frac{3-\sqrt{5}}{2}, \frac{\sqrt{5}-1}{2}]$.

Finally, we will prove that user's strategy combination 6 never occurs at the PBE if $q_+ \in [\frac{3-\sqrt{5}}{2}, \frac{\sqrt{5}-1}{2}]$ (strategy combination 5's never occurring at the PBE can be proved by similar method and we thus omit here).  According to \eqref{A3}, we obtain that the user with bias $b=+$ always deviates from strategy combination 6 to 3 in the regime of $q_+ \in [\frac{3-\sqrt{5}}{2}, \frac{\sqrt{5}-1}{2}]$. Thus, strategy combination 6 is not a PBE for $q_+ \in [\frac{3-\sqrt{5}}{2}, \frac{\sqrt{5}-1}{2}]$. We then finish the proof.

\begin{table*}
    \begin{align}\label{A-L1}
     &\Bar{L}_{Bayesian} = \nonumber \\
     &\begin{cases}
      \frac{q_+ (1-p_H) p_H (\mu _H-\mu
   _L)^2}{p_H + (1-p_H)q_+}, &\!\!\!\!\!\!\text{if} \ q_+ \leq \frac{3-\sqrt{5}}{2} \ \text{and} \ p_H \leq p_{1, H}, \\
   \frac{(1-q_+) (1-p_H) p_H (\mu _H-\mu
   _L)^2}{1-q_+ p_H}, &\!\!\!\!\!\!\text{if} \ q_+ \geq \frac{\sqrt{5}-1}{2} \ \text{and} \ p_H \geq p_{1, L}, \\
   \max\{\frac{q_+ (1-p_H) p_H (\mu _H-\mu
   _L)^2}{p_H + (1-p_H)q_+}, \frac{(1-q_+) (1-p_H) p_H (\mu _H-\mu
   _L)^2}{1-q_+ p_H}\}, &\!\!\!\!\!\!\text{otherwise}.
    \end{cases}
\end{align}
\end{table*}

\section{PBE Analysis for Large Regime of Positively-biased Probability in the One-User Case}

\begin{table}[t]
\caption{PBE versus high-type PDF probability $p_H$ in large positively-biased probability $q_+$$>$$\frac{\sqrt{5}-1}{2}$ in the one-user case.}
\label{t1:prop-3}
% \vskip 0.15in
\begin{center}
\begin{small}
\begin{sc}
\begin{tabular}{|l|l|}
\hline
$p_H$ regime & \multicolumn{1}{c|}{ PBE Results} \\
\hline
\hline
% & PBE \\
% \midrule
Small 
& \multicolumn{1}{c|}{ PBE.1 (with SC.1 in Lemma~\ref{L11})}  \\
% \cdashline{2-2}
$p_H \!\!\in\![0, p_{1, L}]$ & $b$=- user: blindly low-type messaging  \\
& \multicolumn{1}{c|}{$m^*(S|b \!= \!-) \!= \!L$, $\forall S \in \{L, H\}$.}\\ 
% \cdashline{2-2}
 &$b$=+ user: honest messaging\\
& \multicolumn{1}{c|}{$m^*(S|b \!=\! +) \!=\! S$, $\forall S \in \{L, H\}$.} \\
% & $s \in \{-1, 1\}$;\\
% \cdashline{2-2}
&The Platform's inference $a^*(m)$ in \eqref{a-1}.  \\
\cline{2-2} 
&  \multicolumn{1}{c|}{ PBE.2 (with SC.2 in Lemma~\ref{L11})} \\
% \cdashline{2-2}
 & $b$=- user: honest messaging \\ 
& \multicolumn{1}{c|}{$m^*(S|b \!= \!-) \!= \! S$, $\forall S \in \{L, H\}$.}\\
% \cdashline{2-2}}
& $b$=+ user: blindly high-type messaging \\
& \multicolumn{1}{c|}{$m^*(S|b \!=\! +) \!=\! H$, $\forall S \in \{L, H\}$.} \\
% \cdashline{2-2}
% & $s \in \{-1, 1\}$;\\
&The Platform's inference  $a^*(m)$ in \eqref{a-3}.  \\
\cline{2-2} 
& \multicolumn{1}{c|}{ PBE.3 (with SC.3 in Lemma~\ref{L11})}  \\
% \cdashline{2-2}
 &$b$=- user: blindly high-type messaging  \\
 & \multicolumn{1}{c|}{$m^*(S|b \!= \!-) \!= \! H$, $\forall S \in \{L, H\}$.}\\
% \cdashline{2-2}
 &$b$=+ user: reversed messaging \\
& \multicolumn{1}{c|}{$m^*(S=L|b \!= \!+) \!= \!H$,} \\
& \multicolumn{1}{c|}{$m^*(S=H|b \!=\! +) \!=\! L$.} \\
% \cdashline{2-2}
&The Platform's inference $a^*(m)$ in \eqref{a-2}.  \\
\cline{2-2}
& \multicolumn{1}{c|}{ PBE.4 (with SC.4 in Lemma~\ref{L11})}  \\
% \cdashline{2-2}
 & $b$=- user: reversed messaging  \\
 & \multicolumn{1}{c|}{$m^*(S=H|b \!= \!-) \!= \!L$,}\\
& \multicolumn{1}{c|}{$m^*(S=L|b \!= \!-) \!= \!H$.} \\
% \cdashline{2-2}
 &$b$=+ user: blindly low-type messaging \\
& \multicolumn{1}{c|}{$m^*(S|b \!=\! +) \!=\! L$, $\forall S \in \{L, H\}$.} \\
% \cdashline{2-2}
&The Platform's inference $a^*(m)$ in \eqref{a-4}.  \\
\hline
Large &  \multicolumn{1}{c|}{ PBE.1 (with SC.1 in Lemma~\ref{L11})}  \\
\cline{2-2}
$p_H \!\!\in\![p_{1, L}, 1] $ & \multicolumn{1}{c|}{ PBE.3 (with SC.3 in Lemma~\ref{L11})} \\
% & and honest messaging $b$=+. \\
% \cline{2-2} 
% & PBE.2 with strategy combination 2 of  \\
%  & blindly high-type messaging $b$=- \\
%  & and reversed messaging $b$=+. \\
\hline
\end{tabular}
\end{sc}
\end{small}
\end{center}
% \vskip -0.1in
\end{table}

\begin{proposition}\label{prop-3} 
 If the user has large positively-biased probability (i.e., $q_+ > \frac{\sqrt{5}-1}{2}$), we have $p_{1,L}$$\in$$(0,1]$. Non-unique PBEs are given in closed-form in Table~\ref{t1:prop-3}. 
\end{proposition}
\begin{proof}
According to Lemma~\ref{L11}, we know that only user's strategy combinations 1-6 may occur at the PBE. First, we will prove that user's strategy combination 2 occurs at the PBE if $p_H \leq p_{1,L}$ and $q_+ \geq \frac{\sqrt{5}-1}{2}$ (strategy combination 4's occurring at the PBE can be proved by similar method and we thus omit here). We have already shown in Appendix~\ref{A-B} that the only possible deviation from strategy combination 2 is to strategy 6 with the user's bias $b=+$, which occurs if $p_H \leq p_{1,L}$ and $q_+ \geq \frac{\sqrt{5}-1}{2}$ according to \eqref{A3}. Thus, strategy combination 2 is a PBE for $p_H \leq p_{1,L}$ and $q_+ \geq \frac{\sqrt{5}-1}{2}$.

Next, we will prove that user's strategy combination 1 always occurs at the PBE if $q_+ \geq \frac{\sqrt{5}-1}{2}$ (strategy combination 3's occurring at the PBE can be proved by similar method and we thus omit here). We have already shown in Appendix~\ref{A-B} that the only possible deviation from strategy combination 1 is to strategy 5 with the user's bias $b=-$, which still never occurs if $q_+ \geq \frac{\sqrt{5}-1}{2}$ 
% according to \eqref{A6}
. Thus, the user with bias $b=-$ or $b=+$ never deviates from strategy combination 1 and it is a PBE for $q_+ \geq \frac{\sqrt{5}-1}{2}$.

Finally, we will prove that user's strategy combination 5 never occurs at the PBE if $q_+ \geq \frac{\sqrt{5}-1}{2}$ (strategy combination 6's never occurring at the PBE can be proved by similar method and we thus omit here).  
% According to \eqref{A6}, 
We obtain that the user with bias $b-=1$ always deviates from strategy combination 6 to 3 in the regime of $q_+ \geq \frac{\sqrt{5}-1}{2}$. Thus,  strategy combination 5 is not a PBE for $q_+ \geq \frac{\sqrt{5}-1}{2}$. We then finish the proof.
\end{proof}

\section{Proof of Proposition~\ref{prop-4}}

According to Propositions~\ref{prop-1}-\ref{prop-3}, we derive the platform's expected cost in \eqref{u-r} under PBE.1-PBE.4 as follows: 
\begin{align}
    &\bar{c}^{1,*} = \bar{c}^{3,*} \label{pbe-12-u-r} \\
    &=   \frac{q_+ (1-p_H) p_H (\mu _H-\mu
   _L)^2}{p_H + (1-p_H)q_+}+(1-p_H)\sigma_L^2+p_H
   \sigma_H^2, \nonumber \\
   &\bigg( q_+ \leq \frac{\sqrt{5}-1}{2} \bigg) \ \text{or} \ \bigg(  p_H \leq p_{1,L} \ \text{and} \  q_+ > \frac{\sqrt{5}-1}{2} \bigg), 
   \nonumber 
\end{align}
\begin{align}
   & \bar{c}^{2,*} = \bar{c}^{4,*} \label{pbe-34-u-r} \\
   &=  \frac{(1-q_+) (1-p_H) p_H (\mu _H-\mu
   _L)^2}{1-q_+ p_H}+(1-p_H)\sigma_L^2+p_H
   \sigma_H^2, \nonumber \\
   &\bigg( q_+ \geq \frac{3-\sqrt{5}}{2} \bigg) \ \text{or} \ \bigg(  p_H \geq p_{1,H} \ \text{and} \  q_+ \leq \frac{3-\sqrt{5}}{2} \bigg). \nonumber
\end{align}
After comparing the above expected costs, we obtain $\bar{L}_{Bayesian}$ as shown in \eqref{A-L1}, which can be simplified as \eqref{1}. By comparing system loss $\Bar{L}_2$ in \eqref{u-r-b-1} and $\bar{L}_{Bayesian}$ in \eqref{1}, we have 
\begin{align*}
    &\Delta L(\mu_L, \mu_H, q_+, p_H) = \Bar{L}_2 - \bar{L}_{Bayesian} = \\ 
    & \begin{cases}
        \frac{(1-q_+) (1-p_H) p_H^2 (\mu _H-\mu
   _L)^2}{p_H + (1-p_H)q_+}, &\\   &\!\!\!\!\!\!\!\!\!\!\!\!\!\!\!\!\!\!\!\!\!\!\!\!\!\!\!\!\!\!\!\!\!\!\!\!\text{if} \ p_H \leq \min\{\max\{\frac{q_+^2}{2q_+^2-2q_++1}, p_{1, H}\}, p_{1,L}\},  \\  
\frac{q_+ (1-p_H)^2 p_H (\mu _H-\mu
   _L)^2}{1-q_+ p_H}, & \\  &\!\!\!\!\!\!\!\!\!\!\!\!\!\!\!\!\!\!\!\!\!\!\!\!\!\!\!\!\!\!\!\!\!\!\!\!\text{if} \   p_H \geq \min\{\max\{\frac{q_+^2}{2q_+^2-2q_++1}, p_{1, H}\}, p_{1,L}\}.
    \end{cases} \nonumber 
\end{align*}
Note that we have $\Delta L \geq 0$ due to $p_H, q_+ \in [0, 1]$. The upper bound can be obtained by checking derivatives on $p_H$ and $q_+$, sequentially. We take $\Delta L=\frac{(1-q_+) (1-p_H) p_H^2 (\mu _H-\mu
   _L)^2}{p_H + (1-p_H)q_+}$ in the regime of $p_H \leq \min\{\max\{\frac{q_+^2}{2q_+^2-2q_++1}, p_{1, H}\}, p_{1,L}\}$ for example. Denote $\Bar{q}_1 \in (0, 1)$ as the unique solution to 
   \begin{align*}
       4\Bar{q}_1^3-7\Bar{q}_1^2+5\Bar{q}_1-1=0,
   \end{align*}
   $\Bar{q}_2 \in (0, 1)$ as the unique solution to
   \begin{align*}
       2\Bar{q}_2^4-8\Bar{q}_2^3+10\Bar{q}_2^2-7\Bar{q}_2+2=0,
   \end{align*}
   and $\Bar{q}_3 \in (0, 1)$ as the unique solution to
   \begin{align*}
       4\Bar{q}_3^4+4\Bar{q}_3^3-6\Bar{q}_3^2-4\Bar{q}_3+3=0.
   \end{align*}
   After checking derivatives on $p_H$, we have $p_H^*(q_+)$ to maximize $\Delta L$ as follows:
   \begin{align*}
      p_H^*(q_+) = \begin{cases}
          p_{1,H}, &\text{if} \ q_+ < \Bar{q}_1, \\
          \frac{q_+^2}{2q_+^2-2q_++1}, &\text{if} \ q_+ \in (\Bar{q}_1, \Bar{q}_2], \\
          \frac{\sqrt{1+8q_+}+1-4q_+}{4(1-q_+)}, &\text{if} \ q_+ \in (\Bar{q}_2, \Bar{q}_3], \\
          p_{1,L}, &\text{if} \ q_+ > \Bar{q}_3.
      \end{cases}
   \end{align*}
   After substituting the above $p_H^*(q_+)$ in $\Delta L$ and checking derivatives on $q_+$, we obtain that $q_+^*=0$ with $p_H^* = p_{1,H}$ and $\Delta L = (\sqrt{5}-2)(\mu_H-\mu_L)^2$. Upper-bound on $\Delta L$ in the other regime $p_H \geq \min\{\max\{\frac{q_+^2}{2q_+^2-2q_++1}, p_{1, H}\}, p_{1,L}\}$ can be obtain similarly as $\Delta L = (\sqrt{5}-2)(\mu_H-\mu_L)^2$ with $q_+^*=1$ with $p_H^* = p_{1,L}$. We then finish the proof.

\section{Extension to Sequentially-Messaging Users with Asymmetric Messaging Strategies}

% \subsection{Extension on Sequentially-Messaging Users}

In this subsection, we extend to consider the scenario where users sequentially message to the platform with asymmetric messaging strategies. For ease of exposition, we focus on the case of two users, where user 1 messages first and user 2 observing user 1's message then messages to the platform. Note that user 2 has more messaging strategies given his observation on user 1's message. We manage to prove that our PBE results in Proposition~\ref{prop-11} are stable with such sequentially-messaging users of asymmetric strategies.

\begin{proposition}\label{prop-VC}
    In the case of two sequentially-messaging users with asymmetric messaging strategies, we have PBEs are same as Table~\ref{t2-p6} for $N=2$. 
\end{proposition}
% The proof of Proposition~\ref{prop-VC} is given in Appendix F of the supplementary material of this TMC submission. 

\begin{proof}

    Thanks to Proposition~\ref{prop-11}, we can restrict our attentions to strategy combinations 1-4 since the others never occur at the PBE even with simultaneously-messaging users. In the following we will prove that strategy combination 1 is stable to occur at the PBE (the remaining 5 combinations can be proved with similar methods and we thus omit here).

Regarding strategy combination 1, since user 1 is the first to message, he cannot observe user 2's message. Thus, he with either bias will not deviate from strategy combination 1 given user 2's strategy fixed. It is natural and enough to prove that user 2 with either bias will not deviate from strategy combination 1. Next, we assume that user 1 and the negatively-biased user 2 still follow strategy combination 1, and begin to analyze the positively-biased user 2's deviating strategy candidates, which are as follows ($j \in \{L, H\}$):
\begin{itemize}
    \item user 2's strategy $\hat{1}$: 
    $m_2(S = H, m_1 = H|b_2 = +) = H$,
    
    $m_2(S = H, m_1 = L|b_2 = +) = L$,
    
    and $m_2(S = L, m_1 = L|b_2 = +) = L$,
    \item user 2's strategy $\hat{2}$:
    $m_2(S = H, m_1 = H|b_2 = +) = H$,
    
    $m_2(S = H, m_1 = L|b_2 = +) = L$, 
    
    and $m_2(S = L, m_1 = L|b_2 = +) = H$, 
    \item user 2's strategy $\hat{3}$: 
    $m_2(S = H, m_1 = H|b_2 = +) = L$,
    
    $m_2(S = H, m_1 = L|b_2 = +) = H$,
    
    and $m_2(S = L, m_1 = L|b_2 = +) = H$, 
    \item user 2's strategy $\hat{4}$: 
    $m_2(S = H, m_1 = H|b_2 = +) = L$, 
    
    $m_2(S = H, m_1 = L|b_2 = +) = H$, 
    
    and $m_2(S = L, m_1 = L|b_2 = +) = L$.
\end{itemize}
Following similar steps as shown in \eqref{ul'-2} and \eqref{ul-2}, we obtain the positively-biased user 2's expected utility under strategies $\hat{1}$-$\hat{4}$ are, respectively,
\begin{align}
    &\bar{u}_{\hat{1}}(b_2=+) = p_H\mu_H+(1-p_H)\mu_L, \label{u2-1)} \\
    &\bar{u}_{\hat{2}}(b_2=+) = 
    \frac{p_H^2  (\mu_H - \mu_L)}{2} + \frac{p_H  (\mu_H - \mu_L)}{2}  + \mu_L, \label{u2-2)}  \\
    &\bar{u}_{\hat{3}}(b_2=+) = p_H\mu_H+(1-p_H)\mu_L, \label{u2-3)} \\
    % \end{align}
    % \begin{align}
    &\bar{u}_{\hat{4}}(b_2=+) = 
    p_H \mu_H  + (1 - p_H) \frac{p_H (\frac{1}{2})^2 \mu_H + (1 - p_H) \mu_L}{p_H (\frac{1}{2})^2 + (1 - p_H)}. \label{u2-4)} 
\end{align}
To show the positively-biased user 2 will not deviate from strategy combination 1 to strategy $\hat{1}$ according to \eqref{u2-1)} and \eqref{u2-12}, we have
\begin{align*}
    \bar{u}_{\hat{1}}(b_2=+) - \bar{u}_{1}(b_2=+) 
    = \frac{(-1 + p_H) p_H \frac{1}{2^2} (\mu_H - \mu_L)}{1 - p_H (2 - \frac{1}{2} ) \frac{1}{2} } \leq 0
\end{align*}
for $p_H \in [0, 1]$. To show the positively-biased user 2 will not deviate from strategy combination 1 to strategy $\hat{2}$ according to \eqref{u2-2)} and \eqref{u2-12}, we have
\begin{align*}
    &\bar{u}_{\hat{2}}(b_2=+) - \bar{u}_{1}(b_2=+) =\\
    & -\frac{(-1 + p_H) p_H (-1+p_H \frac{1}{2})(2-\frac{3}{2} +\frac{1}{2^2}) (\mu_H - \mu_L)}{1 - p_H (2 - \frac{1}{2}) \frac{1}{2}} \leq 0
\end{align*}
for $p_H \in [0, 1]$.  To show the positively-biased user 2 will not deviate from strategy combination 1 to strategy $\hat{3}$ according to \eqref{u2-3)} and \eqref{u2-12}, we have
\begin{align*}
    \bar{u}_{\hat{3}}(b_2=+) - \bar{u}_{1}(b_2=+) 
    = \frac{(-1 + p_H) p_H \frac{1}{4} (\mu_H - \mu_L)}{1 - p_H (2 - \frac{1}{2}) \frac{1}{2}} \leq 0
\end{align*}
for $p_H \in [0, 1]$. To show the positively-biased user 2 will not deviate from strategy combination 1 to strategy $\hat{4}$ according to \eqref{u2-4)} and \eqref{u2-12}, we have
\begin{align*}
    \bar{u}_{\hat{4}}(b_2=+) - \bar{u}_{1}(b_2=+) = 0.
\end{align*}
In summary, we prove that strategy combination 1 is stable for the positively-biased user 2. By using similar methods as above, we can also prove that strategy combination 1 is stable for the negatively-biased user 2. We then finish the proof.
\end{proof}

% Users' negative and positive biases well offset each other at the PBE for the platform to learn and act. Users 1 and 2 then exhibit symmetric strategy with same bias at the PBE.

According to Proposition~\ref{prop-VC}, we obtain that PBE of two sequentially-messaging users is same as simultaneously-messaging users. The system loss is thus same as \eqref{N} with $N$=2.

\section{Extension to Biased Messaging under Multiple Service Quality Levels}

Recall that only two possible PDF types $\{\phi_L(\theta), \phi_H(\theta)\}$ are considered in the system model in Section II before. Now we extend biased messaging for $N\geq2$ users under multiple PDF types to include another medium PDF type $M$ in the distribution space $\{\phi_L(\theta),  \phi_M(\theta), \phi_H(\theta)\}$, with the mean of the medium PDF $\mu_M$=$\frac{\mu_H+\mu_L}{2}$. Denote set $I_j$=$\{i|$$m_i$=$j$,1$\leq$$i$$\leq$$N\}$ to include all users for messaging $j$$\in$$\{L, M, H\}$. Note that each biased user now has $3^3$=27 strategy candidates, making the analysis more involved. Nonetheless, by following a similar analysis as shown in Section IV, we manage to solve closed-form PBEs to extend from Proposition 1 in Section IV-C.

\renewcommand{\theproposition}{8}
\begin{proposition}\label{prop-A-8}
 Given multiple users $N$$\geq$2 and three distribution types in $\{\phi_L(\theta),  \phi_M(\theta), \phi_H(\theta)\}$ with a uniform probability distribution of $\{\frac{1}{3}, \frac{1}{3}, \frac{1}{3}\}$, we have the following PBEs:
 \begin{itemize}
     \item PBE.1: Each negatively-biased user $i$ blindly sends low-type message  (i.e., $m_i(S|b_i$=$-)$=$L$) and each positively-biased user $i$ honestly messages his observed type 
    (i.e., $m_i(S|b_i$=$+)$=$S$, $\forall S \in \{L, M, H\}$). The platform's equilibrium rating action is 
    \begin{align*}
        a^*(\{m_i\}_{i=1}^N) = \begin{cases}
            \mu_H, &\text{if} \ |I_H| \geq 1 , \\
            \frac{\mu_H+\mu_L}{2}, &\text{if} \ |I_M| \geq 1, \\
            \frac{\frac{\mu_H}{2^N}+\frac{\mu_H+\mu_L}{2^{N+1}}+\mu_L}{\frac{1}{2^N}+\frac{1}{2^N}+1}, &\text{otherwise}.
        \end{cases}
    \end{align*}
    \item PBE.2: Each negatively-biased user $i$ honestly messages his observed type  (i.e., $m_i(S|b_i$=$-)$=$S$) and  each positively-biased user $i$ blindly sends high-type message (i.e., $m_i(S|b_i$=$+)$=$H$, $\forall S \in \{L, M, H\}$).
    The platform's equilibrium rating action is 
    \begin{align*}
        a^*(\{m_i\}_{i=1}^N) = \begin{cases}
            \mu_L, &\text{if} \ |I_L| \geq 1 , \\
            \frac{\mu_H+\mu_L}{2}, &\text{if} \ |I_M| \geq 1, \\
            \frac{\frac{\mu_L}{2^N}+\frac{\mu_H+\mu_L}{2^{N+1}}+\mu_H}{\frac{1}{2^N}+\frac{1}{2^N}+1}, &\text{otherwise}.
        \end{cases}
    \end{align*}
 \end{itemize}
\end{proposition}
 Note that our Bayesian game theoretic approach still manages to persuade positively-biased users for honest messaging in PBE.1 and negatively-biased users for honest messaging in PBE.2. Given PBE.1-2 in Proposition~\ref{prop-A-8}, we examine the system performance of our Bayesian game theoretic learning approach and analyze the platform's system loss to extend from Theorem 1 in Section IV-C.
 \renewcommand{\thetheorem}{3}
\begin{theorem}\label{prop-A-9}
Given multiple $N\geq2$ users and three distribution types in $\{\phi_L(\theta),  \phi_M(\theta), \phi_H(\theta)\}$ with a uniform probability distribution of $\{\frac{1}{3}, \frac{1}{3}, \frac{1}{3}\}$, our Bayesian game theoretic approach in Proposition~\ref{prop-A-8} has the following system loss:
\begin{align*}
    \bar{L}_{Bayesian} =  
    \frac{(1+5\cdot 2^N)}{3(2+2^N)2^{2+N}}(\mu_H-\mu_L)^2,
\end{align*}
which decreases with the user number $N$ with $\lim_{N\to\infty}\bar{L}_{Bayesian} = 0$. The incurred system loss is still obviously less than $\bar{L}_1$ of the benchmark 1 and $\bar{L}_2$ of the benchmark 2. 
\end{theorem}
Regarding multiple PDF types $n>3$, note that each biased user has $n^n$ strategy candidates, making it even more challenging to obtain closed-form PBE. Our analysis in Section IV still applies yet it may be intractable due to the exponentially-increased complexity. It requires us to adopt new approximation methods to determine biased users' messaging strategies.

\end{document}